\pdfoutput=1
\documentclass[a4paper,11pt,reqno]{amsart}
\usepackage{amsmath}
\usepackage{amssymb}
\usepackage{bm}
\usepackage{graphicx}
\usepackage{epsfig}
\usepackage[latin1]{inputenc}
\usepackage{amsfonts}
\usepackage{amsthm}
\usepackage{color}
\usepackage{newlfont}
\usepackage[round,comma]{natbib}
\usepackage{latexsym}
\usepackage{subfig}
\usepackage[colorlinks,citecolor=blue,urlcolor=blue]{hyperref}

\newtheorem{theorem}{Theorem}

\newtheorem{lemma}[theorem]{Lemma}

\theoremstyle{definition}

\newtheorem{preass}{Definition}

\input{ee.sty}

\begin{document}
\title[Exploring HP's designs using GSA]{Exploring Hoover and Perez's
experimental designs \\
using global sensitivity analysis}
\author[Becker, Paruolo \& Saltelli]{William Becker$^\circ$}
\author[]{Paolo Paruolo$^{\circ }$}
\author[]{Andrea Saltelli$^{\circ }$}
\thanks{$^{\circ }$
Econometrics and Applied Statistics Unit,
European Commission, DG JRC IPSC,
Via E.Fermi 2749, I-21027 Ispra (VA), Italy;
emails: \texttt{william.becker@jrc.ec.europa.eu}, \texttt{paolo.paruolo@jrc.ec.europa.eu}, \texttt{andrea.saltelli@jrc.ec.europa.eu}.
}
\date{\today.
}

\begin{abstract}
This paper investigates variable-selection procedures in regression that
make use of global sensitivity analysis. The approach is combined with
existing algorithms and it is applied to the time series regression designs
proposed by Hoover and Perez. 
A comparison of an algorithm employing global sensitivity analysis and the
(optimized) algorithm of Hoover and Perez shows that the former
significantly improves the recovery rates of original specifications.%
\end{abstract}

\maketitle



\begingroup
\endgroup


\section{Introduction}

Model selection in regression analysis is a central issue, both in theory
and in practice. A partial list of statistical fields with a non-empty
intersection with model selection includes multiple testing, see e.g. \cite%
{Romano2005} and \cite{Bittman2009}, pre-testing, see \cite{Leeb2006},
information criteria, see \cite{Hjort2003} and \cite{liu11}, model selection
based on Lasso, see e.g. \cite{Brunea2008}, model averaging, see \cite%
{Claeskens2003}, stepwise regression, see \cite{Miller2002}, risk inflation
in prediction, see \cite{Foster1994}, directed acyclic graphs and causality
discovery, see e.g. \cite{Freedman1999}.\footnote{%
Model selection is also associated with current rules of thumb on the
maximum number of regression parameters to consider. This literature appears
to have been initiated by \cite{Freedman1983}, who considered the case of a
first screening regression with 50 regressors and 100 data points, where
regressors that are significant at 25\% significance level are kept in a
second regression. Freedman showed that the second regression is troublesome
when one acts as if the screening regression had not been performed and the
ratio of number of observations to number of regressors in the screening
regression is kept in a fixed proportion as the number of observations
diverges. This study was followed by \cite{Freedman1989, Freedman1992}, who
defined the rule of thumb that the ratio of the number of observations per
regressor should be at leat equal to 4; this rule is included in \cite%
{Harrell2001}, who suggested to have it at least equal to 10.}

Model choice is also of primary concern in many areas of applied
econometrics, as witnessed for example by the literature on growth
regression, see e.g. \cite{Sala1997}. Controlling for the right set of
covariates is central in the analysis of policy impact evaluations; this is
embodied in the assumption of unconfoundedness, see e.g. \cite{Imbens2009}.
In economic forecasting, model selection is the main alternative to model
averaging, see e.g. \cite{Hjort2003}.

The analysis of the effects of pre-testing on parameter estimation has a
long tradition in econometrics, see \cite{Danilov2004} for a recent account;
in this context \cite{Magnus1999} and co-authors proposed the weighted
average least squares estimator (WALS), and compared it with model averaging
for growth empirics, see \cite{Magnus2010}.

Model selection is a major area of investigation also in time-series
econometrics, see e.g. \cite{Phillips1997,Phillips2003}. The so-called
London School of Economics (LSE) methodology has played a prominent role in
this area, advocating the \emph{general-to-specific} (GETS) approach to
model selection, see \cite{Hendry2005}, \cite{castle_evaluating_2011} and
references therein. In a widely cited paper,
\cite{hoover99} (hereafter HP) `mechanized' -- i.e. translated -- the GETS
approach into an algorithm for model selection; they then tested the
performance of the HP\ algorithm on a set of time-series regression designs,
constructed along the lines of \cite{Lovell1983}.

Model selection is also related to the issue of regression coefficients'
robustness (i.e. lack of sensitivity) to the omission/inclusion of
additional variables. \cite{leamer1983} has proposed extreme bound analysis,
i.e. to report the range of possible parameter estimates of the coefficient
of interest when varying the additional regressors included in the analysis,
as an application of sensitivity analysis to econometrics. Other
applications of sensitivity analysis to econometrics include the local
sensitivity to model misspecification developed in \cite{Magnus2007a} and
\cite{Magnus2007b}.\footnote{%
They show that local sensitivity measures provide complementary information
with respect to standard diagnostic tests for misspecification, i.e. that
the two types of statistics are asymptotically independent. In SA a local
measure of sensitivity is one focused on a precise point in the space of the
input factor, a.g. a partial derivative of the output versus the input. With
a global measure of sensitivity the influence of a given input on the output
is averaged both on the distribution of the input factor itself and on the
distributions of all the remaining factors, see \cite{saltelli1993}.}

Sensitivity analysis originated in the natural sciences, and it is generally
defined as `the study of how the uncertainty in the output of a mathematical
model or system (numerical or otherwise) can be apportioned to different
sources of uncertainty in its inputs', see \cite{Saltelli2002ra}.
Deterministic models based on knowledge of the physical laws governing the
system are usually fruitfully applied in physical sciences. \cite{Box2007}\
advocate their use\ (in combination with statistical models) because they
(i) contribute to the scientific understanding of the phenomenon under
study, (ii) provide a better basis for extrapolation with respect to
empirical models\ (iii) they tend to be parsimonious (i.e, frugal) in the
use of parameters. The combined use of deterministic and stochastic models
is also advocated in other non-experimental fields, such as in environmental
modeling, see \cite{Young1996}.

Despite several uses of sensitivity in econometrics, the present authors are
not aware of systematic applications of the techniques of Global Sensitivity
Analysis,\ GSA, see \cite{saltelli1993}, to the problem of model selection
in regression. The present paper attempts a first experimental exploration
of the possible application of GSA to model selection in time-series\
regression. Here we aim to answer the question: \textquotedblleft Can GSA
methods help in model selection in practice?\textquotedblright. This
question is answered in the affirmative, using the `total sensitivity index'
to rank regressors' importance in order to construct relevant subsets of
models.

For simplicity and in order to increase replicability of our exploration, we
have chosen to compare new tools and old via simple Monte Carlo (MC) methods.%
\footnote{%
See however analytical results on the properties of the GSA-based algorithm
in Appendix A.} We have chosen to replicate the original search algorithm in
HP as a benchmark, and to compare our GSA algorithm with the original HP
algorithm. Because the purpose of the paper is to investigate the
contribution of GSA methods, we have abstained from the implementation of
any other algorithms except the original -- albeit optimized -- HP algorithm
as a benchmark.

The choice of the design of experiments in HP reflects the current practice
in single-equation, time-series econometric models; these consist in a
possibly dynamic regression equation with exogenous variables, where the
exogenous variables are fixed across experiments and are taken from
real-world, stationary, macroeconomic time series. While HP's designs are
supposed to represent prototypical configurations in time-series
econometrics, they contain by construction only a small subset of possible
situations encountered in econometric time-series applications. As such, it
is like a single planet in a galaxy.\

As forbidding as the exploration of a galaxy is (at least with current
means), so is the attempt to investigate all regression designs. In this
paper we have hence decided to explore a limited part of this galaxy -- a
single planet -- namely HP's experimental designs.%

Several papers appear to have applied other methods to HP's designs, see
\cite{Hendry1999}, \cite{castle_evaluating_2011}.\ The choice of HP's
designs and of the HP algorithm as benchmark allows to compare performances
in our paper with others reported in the literature. The designs in HP's
designs include data generating processes (DGPs) of varying degree of
difficulty (for model search algorithms) and a single sample size of 139
time periods, close to the ones available in typical macroeconomic
applications with quarterly data.

The features of HP's designs prompt a number of considerations. First,
because sample size is limited and fixed, consistency of model-selection
algorithms cannot be the sole performance criterion. In this light, it would
be helpful to be able to describe the complete finite sample properties of
model-selection algorithms for HP's designs; the MC approach taken in the
paper allows to do this.

Secondly, some of the DGPs in HP's designs are characterized by a low
signal-to-noise ratio for some coefficients; we call the corresponding
regressors `weak'. This situation makes it very difficult for statistical
procedures to discover if the corresponding regressors should be included or
not. This raises the question of how to measure selection performance in
this context.

In this paper we observe that, in the case of weak regressors, one can
measure performance of model-selection algorithms also with respect to a
simplified DGP, which contains the subset of regressors with sufficiently
high signal-to-noise ratio; we call this the `Effective DGP', EDGP. In this
paper we make the definition of the EDGP operational using the
`parametricness index' recently introduced in \cite{liu11}.

Overall, results point to the possible usefulness of GSA methods in model
selection algorithms. In particular, GSA methods seem to complement existing
approaches, as they give a way to construct viable search paths (via
ordering of regressors) that are complementary to the ones based on $t\,\,$%
ratios. When comparing optimized algorithms, the GSA method appears to be
able to reduce the failure rate in recovering the EDGP from 5\% to 1\%
approximately -- a five-fold reduction. When some of the regressors are
weak, the recovery of exact DGPs does not appear to be improved by the use
of GSA methods.

Selection of a subset of all possible models from the space of all submodels
is one of the critical aspects also for model selection based on information
criteria, see Section 5.2. in \cite{Hansen1999}. A similar remark applies
for multi-model inference procedures, see e.g. \cite{Burnham2002}. The
results obtained in this paper show that GSA methods have potential to make
these methods operational. Due to space limitations, we leave the analysis
of these extensions %
to future research.

The rest of the paper is organized as follows. In Section \ref%
{sec_problem_definition} we define the problem of interest and introduce
HP's data generating processes and the HP algorithm. Section \ref{sec:GSA}
defines the tools from GSA used in the paper, while Section \ref%
{sec:algorithm} presents the GSA algorithm. Results are reported in Section %
\ref{sec:results}, and Section \ref{sec:conclusions} concludes. Large-sample
properties of the orderings based on the GSA algorithm are discussed in
Appendix A. A discussion about the identifiability of DGPs and the
definition of EGDP is reported in Appendix B. Finally, this paper follows
the notational conventions in \cite{abadir02}.

\section{Problem definition}

\label{sec_problem_definition}This section presents the setup of the problem
and describes the design of experiments in HP, as well as their algorithm.

\subsection{Model selection in regression}

\label{subsec_problem}Let $n$ be the number of data points and $p$ the
number of regressors in a standard multiple regression model of the form%
\begin{equation}
\vy=\mX_{1}\beta _{1}+\dots \mX_{p}\beta _{p}+\vvarepsilon=\mX\vbeta+%
\vvarepsilon  \label{eq_DGP1}
\end{equation}%
where $\vy=(y_{1},\dots ,y_{n})^{\prime }$ is $n\times 1$, $\mX=(\mX%
_{1},\dots ,\mX_{p})$ is $n\times p$, $\mX_{i}:=(x_{i,1},\dots
,x_{i,n})^{\prime }$ is $n\times 1$, $\vbeta=(\beta _{1},\dots ,\beta
_{p})^{\prime }$ is $p\times 1$ and $\vvarepsilon$ is is a $n\times 1$
Gaussian random vector with distribution $N(\vzeros,\sigma ^{2}\mI_{n})$.
The symbol\ $^{\prime }$ indicates transposition.

Let $\Gamma $ be the set of all $p\times 1$ vectors of indicators
$\vgamma=(\gamma _{1},\dots ,\gamma _{p})^{\prime }$, 
with $\gamma _{i}=0$ or $1$ for $i=1,\dots ,p$, i.e. $\Gamma =\{0,1\}^{p}$.
A submodel of (\ref{eq_DGP1}) (or one specification) corresponds to one
vector $\vgamma\in \Gamma $, where $\gamma _{i}=0$ (respectively 1)
indicates that $\beta _{i}$ is to be estimated equal to 0 (respectively
unrestrictedly).
Note that there are $2^{p}$ different specifications, i.e. $\vgamma$ vectors
in $\Gamma $. When $p=40$ as in HP's designs, the number of specifications $%
2^{p}\approx 1.0995 \cdot 10^{12}$ is very large.

In the following we indicate by $\vbeta_{0} =(\beta _{0,1},\dots ,\beta
_{0,p})^{\prime }$ the true value of $\vbeta$. Define also $\vgamma
_{0}=(\gamma _{0,1},\dots ,\gamma _{0,p})^{\prime }$ with $\gamma
_{0i}=1(\beta _{0,i}\neq 0)$, where $1(\cdot )$ denotes the indicator
function and $\beta _{0,i}$ are the true parameters in (\ref{eq_DGP1}). The
vector of indicators $\vgamma_{0}$ defines the smallest true submodel; this
is called the \emph{Data Generating Process} (DGP) in the following.

The least squares estimator of $\vbeta $ in model $\vgamma$ can be written
as follows:
\begin{equation}
\widehat{\vbeta}_{\vgamma}=\left( \mD_{\vgamma}\mX^{\prime }\mX \mD_{\vgamma%
}\right) ^{+}\mD_{\vgamma}\mX^{\prime }\vy,  \label{eq_OLS}
\end{equation}%
where $\mD_{\vgamma}=\mathrm{diag}(\vgamma)$ is the $p\times p$ matrix with
diagonal elements $\vgamma$ and $\mA^{+}$ indicates the Moore-Penrose
generalized inverse of the matrix $\mA$. The non-zero elements in $\widehat{%
\vbeta}_{\vgamma}$ correspond to the least squares estimates in the submodel
which includes only regressors $\mX_{i}$ for which $\gamma _{i}=1$.

The case of $\vgamma$ equal to $\vones$, a vector with all 1s, is called the
\emph{General Unrestricted Model}, the GUM in HP. The problem of interest
is, given the observed data, to find $\vgamma_{0}$, i.e. to identify the DGP.%
\footnote{%
All empirical models are assumed to contain the constant; this is imposed
implicitly by de-meaning the $\vy$ and $\mX_{i}$ vectors. Hence in the
following, the `empty set of regressors' refers to the regression model with
only the constant.}

In this paper we assume that the model is correctly specified, i.e. that $%
\vgamma_{0}$ is an element of $\Gamma $. This is a common hypothesis in the
regression literature. In econometrics this assumption appears be more
questionable, because of the possibility of relevant omitted variables.
However, we maintain it here 
for reasons of simplicity.

\subsection{HP's designs}

\label{subsec_HPworld}HP's designs are constructed as follows. Following
\cite{Lovell1983}, HP chose a set of 18 major US quarterly macroeconomic
variables. Only two variables considered in \cite{Lovell1983} were discarded
in HP, namely the linear trend and the `potential level of GNP in \$1958',
because they were no longer relevant or available. Unlike in \cite%
{Lovell1983}, HP applied 0, 1 or 2 differences to the data; the order of
differencing was selected by HP in order to obtain stationary variables
according to standard unit root tests, see their Table 1.

The values of these (differenced) 18 major US quarterly macroeconomic series
are then fixed in HP's designs; they are here indicated as $x_{it}^{\ast }$,
where $t=1,\dots ,n$ indicates quarters and $i=1,\dots ,k$, with $k=18$
indexes variables. The values of $y_{t}$ were then generated by the
following scheme
\begin{equation}
y_{t}=\sum_{i=1}^{k}\beta _{i}^{\ast }x_{it}^{\ast }+u_{t}\qquad u_{t}=\rho
u_{t-1}+\varepsilon _{t},  \label{eq_yseries1}
\end{equation}%
where $\varepsilon _{t}$ are i.i.d. $N(0,\sigma ^{2})$. Here $\beta
_{i}^{\ast }$ for $i=1,\dots ,k$ and $\sigma ^{2}$ are known constants,
which define the DGP. In practice $\varepsilon _{t}$s are simulated using a
computer random number generator, $u_{t}$ is then calculated as an
autoregressive series of order 1, AR(1), with coefficient $\rho $. $u_{t}$
is then fed into the equation for $y_{t}$, where $x_{it}^{\ast }$ are kept
fixed and do not change across replications.

It is useful to express \eqref{eq_yseries1} as a special case of %
\eqref{eq_DGP1}. To this end one can substitute $(y_{t}-\sum_{i=1}^{k}\beta
_{i}^{\ast }x_{it}^{\ast })$ in place of $u_{t}$ in the dynamic equation of $%
u_{t}$; one hence finds the following equivalent representation of the DGP
\begin{equation}  \label{eq_yseries2}
y_{t}=\rho y_{t-1}+\sum_{i=1}^{2k}\beta _{i}x_{it}+\varepsilon _{t}
\end{equation}%
where $\beta _{i}=\beta _{i}^{\ast }$ and $x_{it}=x_{it}^{\ast }$ for $%
i=1,\dots ,k$ while $\beta _{i}=-\rho \beta _{i}^{\ast }$ and $%
x_{it}=x_{it-1}^{\ast }$ for $i=k+1,\dots ,2k$. This representation is in
the form \eqref{eq_DGP1}, and the parameters can be estimated as in %
\eqref{eq_OLS}.

Regressions in HP were performed setting the elements $x_{i,t}$ in column $%
\mX_{i}$ equal to variable $x_{it}$ from \eqref{eq_yseries2}, for $i=1,\dots
,2k$ with $2k=36$, and setting the elements $x_{i,t}$ of the remaining
columns $\mX_{i}$ for $i=2k+1,\dots ,p$, i.e. from 37 to 40, equal to the
first, second, third and fourth lag of $y_{t}$. Therefore, 4 lags were
always considered in estimation (even if only one lag was possibly present
under the DGP), and the only part of the $\mX$ that changes across
replications is the last 4 columns.

HP defined 11 experimental designs (DGPs)\ by choosing values for the
parameters $\rho $, $\beta _{i}^{\ast }$ and $\sigma _{\varepsilon }^{2}$.
Table \ref{Table_design} summarizes the chosen parameter values. The choice
of these values was made to reflect the coefficient estimates obtained on US
data, using personal consumption expenditure as dependent variable,
following the rationale in \cite{Lovell1983}. Because they were chosen as
explanatory variables for a consumption equation, not all the macroeconomic
time series were included in the DGP;\ in particular only (the second
differences of the) Government purchases on goods and services $G$ and the
(first differences of the) $M1$ monetary aggregate, and their respective
first lags, were included in the designs.
\begin{table}[tbp] \centering%
\begin{tabular}{lccccccccccc}
\hline
DGP & 1 & 2 & 3$^{\natural}$ & 4 & 5 & 6 & 6A & 6B & 7 & 8 & 9 \\ \hline
&  & \multicolumn{8}{c}{coefficients in DGP} &  &  \\ \hline
$y_{t-1}$ &  & 0.75 & 0.395 &  &  &  &  &  & 0.75 & 0.75 & 0.75 \\
$y_{t-2}$ &  &  & 0.3995 &  &  &  &  &  &  &  &  \\
$G_{t}$ &  &  &  &  & -0.046 & -0.023 & -0.32 & -0.65 &  & -0.046 & -0.023
\\
$G_{t-1}$ &  &  &  &  &  &  &  &  &  & 0.00345 & 0.01725 \\
$M1_{t}$ &  &  &  & 1.33 &  & 0.67 & 0.67 & 0.67 & 1.33 &  & 0.67 \\
$M1_{t-1}$ &  &  &  &  &  &  &  &  & -0.9975 &  & -0.5025 \\ \hline
$\sigma _{\varepsilon }$ & 130 & 85.99 & 0.00172 & 9.73 & 0.11 & 4.92 & 4.92
& 4.92 & 6.73 & 0.073 & 3.25 \\ \hline
\end{tabular}%
\vspace{0.3cm}
\caption{DGPs design. $y_{t-j}$ indicates lags of the dependent variable, $G_{t-j}$ denotes (lags of)
second differences of government purchases of goods and services
and $M1_{t-j}$ indicates (lags of) first differences of M1.\\
$^\natural$: in DGP 3 the regression analysis is performed on
$y_t^\ast=\exp (y_t)$, where $y_t$ is simulated as in \eqref{eq_yseries2}.
}\label{Table_design}%
\end{table}%

\subsection{HP algorithm}

HP proposed an algorithm that aims to provide a close approximation to a
subset of what practitioners of the LSE approach actually do. Here we follow
\cite{Hansen1999} in his description of the HP algorithm.

The HP algorithm can be described by a choice of a triplet $(R,f,\Gamma
_{s}) $ composed of (i) a test procedure $R$, (ii) a measure of fit $f$ and
(iii) a subset $\Gamma _{s}$ of all models $\Gamma $, $\Gamma _{s}\subseteq
\Gamma $. For any model $\vgamma$, the test procedure $R$ is defined as%
\begin{equation}
R(\vgamma)=1(\min_{1\leq \ell \leq v}p_{\ell }\leq \alpha )
\label{eq_test_Stat}
\end{equation}%
where $p_{\ell }$ are the $p$-values of $v$ specification tests and $\alpha $
is the chosen significance level.
Note that $R(\vgamma)=0$ when all $v$ tests do not reject the null, which
corresponds to the hypothesis of correct specification and/or constant
parameters.\footnote{%
The tests are the following:\ (1) Jarque Bera test for normality of
residuals;\ (2) Breusch Pagan residual autocorrelation tests;\ (3)\ Engle's
ARCH test on residuals;\ (4) Chow sample-split parameter stability tests;\
(5) Chow out-of-sample stability test using the first 90\% of observations
versus the last 10\%;\ (6) $F$ test of the restrictions imposed by model $%
\gamma $ versus the GUM. The tests are performed on the first 90\% of
observations during the search.}

HP's measure of fit $f$ is based on the least-square estimate of $\sigma
^{2} $, the regression variance, which equals $\widetilde{\sigma }_{\vgamma%
}^{2}:=\frac{1}{n-k_{\vgamma}}\widehat{\vvarepsilon}_{\vgamma}^{\prime }%
\widehat{\vvarepsilon}_{\vgamma}$, where $k_{\vgamma}$ and $\widehat{%
\vvarepsilon}_{\gamma }$ are the number of regressors and the residuals in
model $\vgamma$. HP's measure of fit is $f(\vgamma)=\widetilde{\sigma }_{%
\vgamma}$, which should be minimized. Finally the subset $\Gamma _{s}$ is
selected recursively, going from general to specific models, starting from
the GUM, $\vgamma=\vones_{p}$; the recursion continues as long as $R(\vgamma%
)=0$. Details on HP's choice of $\Gamma _{s}$ are given in the next section.

Overall the HP algorithm selects a model $\widehat{\vgamma}$ as the
preferred model using the rule%
\begin{equation*}
\widehat{\vgamma}=\underset{\vgamma\in \Gamma _{s}:R(\vgamma)=0}{\arg \min }%
f(\vgamma).
\end{equation*}%
The above description shows that the HP algorithm depends on $\alpha $,
which is a tuning parameter, as well as on the choice of specific path $%
\Gamma _{s}$. For large $n$, \cite{Hansen1999} noted that $\widehat{\vgamma}$
corresponds approximately to minimizing the information criterion\footnote{%
Here the only approximation involved in the large $T$ argument is $\log
(1+k_{\gamma }/T)\approx k_{\gamma }/T$.} $HP(\vgamma)=\log \widehat{\sigma }%
_{\vgamma}^{2}+k_{\vgamma}/n$, where $\widehat{\sigma }_{\vgamma}^{2}:=\frac{%
1}{n}\widehat{\vvarepsilon}_{\vgamma}^{\prime }\widehat{\vvarepsilon}_{%
\vgamma}$ is the ML estimator of $\sigma ^{2}$. This differs from Akaike's
Information Criterion $AIC(\vgamma)=\log \widehat{\sigma }_{\vgamma}^{2}+2k_{%
\vgamma}/n$ and from the Bayesian Information criterion of Schwarz $BIC(%
\vgamma)=\log \widehat{\sigma }_{\vgamma}^{2}+k_{\vgamma}\log (n)/n$ by the
different choice of penalty term.\footnote{%
Remark that information criteria are equivalent to LR testing with a tunable significance level; see for instance \cite{Poetscher1991}.}

\subsection{A subset of models}

The number of models in $\Gamma $ is too large to visit all submodels;\
hence any selection method needs to select at most a subset $\Gamma _{s}$ of
$\Gamma $. This is a critical aspect of the HP algorithm, as well as of any
selection method based e.g. on information criteria, see Section 5.2. in
\cite{Hansen1999} and \cite{Burnham2002}.

In particular, HP select a subset $\Gamma _{s}$ as follows. All paths start
from the GUM regression, and the regressors are ranked in ascending order
according the their $t$-statistics. The 10 lowest variables in this list are
then candidates for elimination; this starts an iterative elimination path.
Each candidate model $\vgamma_{\ast }$ then becomes the current
specification provided $R(\vgamma_{\ast })=0$. In this stage, the first 90\%
of the observations are used in the specification tests. Each search is
terminated when for any choice of regressor the test $R\ $rejects.

At this final stage, the HP algorithm reconsiders all the observations in a
`block search'; this consists in considering the joint elimination of all
the regressors with an insignificant $t$-statistics. If the $R$ tests for
the block search does not reject, the resulting model becomes the terminal
specification. Otherwise, the specification that entered the final stage
becomes the terminal specification. Once all 10 search paths have ended in a
terminal specification, the final specification is the one among these with
lowest $f(\vgamma)=\widetilde{\sigma }_{\vgamma}$.

This paper gives a contribution on the selection of $\Gamma _{s}$,
by defining a GSA-based ordering of regressors.

\subsection{Measures of performance}

\label{subsec_measures}

The performance of algorithms was measured by HP via the number of times the
algorithm selected the DGP as a final specification. Here we describe
measures of performance similar to the ones in HP, as well as additional
ones proposed in \cite{castle_evaluating_2011}.

Recall that $\vgamma_{0}$ is the true set of included regressors and let $%
\widehat{\vgamma}_{j}$ indicate the one produced by a generic algorithm in
replication $j=1,\dots, N_R$. 
Here we define $r_{j}$ to be number of correct inclusions of components$\ $%
in vector $\widehat{\vgamma}_{j}$, i.e. the number of regression indices $i$
for which $\widehat{\gamma }_{j,i}=\gamma _{0,i}=1$, $r_{j}=\sum_{i=1}^{p}1(%
\widehat{\gamma }_{j,i}=\gamma _{0,i}=1)$. Similarly, we let $r_{0}$
indicate the number of true regressors.

We then define the following exhaustive and mutually exclusive categories of
results:

\begin{enumerate}
\item[$C_1$:] exact matches;

\item[$C_2$:] the selected model is correctly specified, but it is larger
than necessary, i.e. it contains all relevant regressors as well as
irrelevant ones;

\item[$C_3$:] the selected model is incorrectly specified (misspecified),
i.e. it lacks relevant regressors.
\end{enumerate}

$C_{1}$ matches correspond to the case when $\widehat{\vgamma}_{j}$
coincides with $\vgamma_{0}$;\ \ the corresponding frequency $C_{1}$ is
computed as $C_{1}=\frac{1}{N_{R}}\sum_{j=1}^{N_{R}}1(\widehat{\vgamma}_{j}=%
\vgamma_{0})$. The frequency of $C_{2}$ cases is given by $C_{2}=\frac{1}{%
N_{R}}\sum_{j=1}^{N_{R}}1(\widehat{\vgamma}_{j}\neq \vgamma_{0},r_{j}=r_{0})$%
. Finally, $C_{3}$ cases are the residual category, and the corresponding
frequency is $C_{3}=1-C_{1}-C_{2}$.\footnote{$C_{1}$ corresponds to Category
1 in HP;\ $C_{2}$ corresponds to Category 2+Category 3$-$Category 1 in HP;\
finally $C_{3}$ corresponds to Category 4 in HP.}

The performance can be further evaluated through measures taken from \cite%
{castle_evaluating_2011}, known as potency and gauge. First the retention
rate $\widetilde{p}_{i}$ of the $i$\textbf{-}th variable is defined as, \ $%
\widetilde{p}_{i}=\frac{1}{N_{R}}\sum_{j=1}^{N_{R}}1(\widehat{\gamma }%
_{j,i}=1)$.$\ $Then, potency and gauge are defined as follows:%
\begin{equation}
\text{potency}=\frac{1}{r_{0}}\sum_{i:\beta _{0,i}\neq 0}\widetilde{p}%
_{i},\qquad \text{gauge}=\frac{1}{p-r_{0}}\sum_{i:\beta _{0,i}=0}\widetilde{p%
}_{i},  \notag
\end{equation}%
where $r_{0}$ indicates the number of true regressors in the DGP.

Potency therefore measures the average frequency of inclusion of regressors
belonging to the DGP, while gauge measures the average frequency of
inclusion of regressors not belonging to the DGP. An ideal performance is
thus represented by a potency value of 1 and a gauge of 0.

In calculating these measures, HP chose to discard MC samples for which a
preliminary application of the battery of tests defined in (\ref%
{eq_test_Stat}) reported a rejection.\footnote{%
We found the empirical percentage of samples that were discarded in this way
was proportional to the significance level $\alpha$. This fact, however, did
not influence significantly the number of $C_1$ catches. We hence decided to
let the HP procedure discard sample as in the original version. For the GSA
algorithm we did not apply any pre-search elimination.} We call this choice
`pre-search elimination' of MC samples.

\subsection{Benchmark}

In this paper we take the performance of HP's algorithm as a benchmark. The
original \texttt{MATLAB} code for HP designs and the HP algorithm was
downloaded from HP's home page.\footnote{\texttt{%
http://www.csus.edu/indiv/p/perezs/Data/data.htm}} The original scripts were
then updated to run on the current version of \texttt{MATLAB}. A replication
of the results in Tables 4, 6 and 7 in HP is reported in the first panel of
Table~\ref{tab:results_HP}, using a nominal significance level of $\alpha
=1\%,5\%,10\%$ and $N_R=10^{3}$ replications. The results do not appear to
be significantly different from the ones reported in HP.

We then noted an incorrect coding in the original HP script for the
generation of the AR series $u_{t}$ in eq. (\ref{eq_yseries1}), which
produced simulations of a moving average process of order 1, MA(1), with MA
parameter $0.75$ instead of an AR(1) with AR parameter $0.75$.\footnote{%
This means that the results reported in HP for DGP 2, 3, 7, 8, 9 concern a
design in which the entertained model is misspecified. The MA process can be
inverted to obtain a AR($\infty $) representation; substituting from the $%
y_{t}$ equation as before, one finds that the DGP contains an infinite
number of lags on the dependent variable and of the $x_{it}^{\ast }$
variables, with exponentially decreasing coefficients. The entertained
regression model with 4 lags on the dependent variable and 2 lags on the $%
x_{it}^{\ast }$ variables can be considered an approximation to the DGP.} We
hence modified the script to produce $u_{t}$ as an AR(1) with AR parameter $%
0.75$; we call this the `modified script'. Re-running the experiments using
this modified script we obtained the results in the second panel in Table~%
\ref{tab:results_HP}; for this set of simulations we used $N_{R}=10^{4}$
replications. Comparing the first and second panels in the table for the
same nominal significance level $\alpha $, one observes a significant
increase in $C_{1}$ catches in DGP 2 and 7. One reason for this can be that
when the modified script is employed, the regression model is
well-specified, i.e. it contains the DGP as a special case.\footnote{%
This finding is similar to the one reported in \cite{Hendry1999}, section 6;
they re-run HP design using PcGets, and they document similar increases in $%
C_{1}$ catches in DGP 2 and 7 for their modified algorithms. Hence, it is
possibile that this result is driven by the correction of the script for the
generation of the AR series.}

The original results in HP and those obtained with the modified script in
Table~\ref{tab:results_HP} document how HP's algorithm depends on $\alpha $,
the significance level chosen in the test $R$ in (\ref{eq_test_Stat}).

\begin{table}[tbp] \centering%
\begin{tabular}{l|rrr|rrr}
\hline
& \multicolumn{3}{c}{original script} & \multicolumn{3}{|c}{modified script}
\\ \hline
DGP & $\alpha =0.01$ & 0.05 & 0.1 & 0.01 & 0.05 & 0.1 \\ \hline
1 & 81.1 & 28.6 & {6.8} & 79.3 & 30.0 & 7.3 \\
2 & 1.2 & 0.0 & {0.0} & 77.0 & 27.0 & 6.9 \\
3 & 71.4 & 27.2 & {9.1} & 71.7 & 27.3 & 6.9 \\
4 & 78.2 & 31.2 & {6.4} & 81.8 & 31.1 & 7.0 \\
5 & 80.9 & 30.1 & {7.4} & 80.7 & 29.9 & 6.4 \\
6 & 0.2 & 1.0 & {0.7} & 0.4 & 0.5 & 0.6 \\
6A & 68.0 & 27.8 & {7.8} & 70.6 & 27.8 & 7.6 \\
6B & 80.8 & 30.7 & {7.8} & 81.1 & 31.4 & 8.0 \\
7 & 23.6 & 4.7 & {0.3} & 75.7 & 26.7 & 7.6 \\
8 & 80.6 & 31.0 & {8.0} & 79.2 & 30.3 & 9.4 \\
9 & 0.1 & 0 & {0} & 0 & 0 & 0 \\ \hline
\end{tabular}%
\vspace{0.5cm}%
\caption{Percentages of Category 1 matches $C_1$ for different values of $\alpha$.
Original script: data generated by the original script, $N_R=10^3$ replications.
The frequencies are not statistically different from the ones reported in HP tables
4, 6, 7. Modified script:
data from modified script for the generation of AR series, $N_R=10^4$ replications.
}\label{tab:results_HP}%
\end{table}%

\section{GSA approach}

\label{sec:GSA}The HP algorithm uses $t$-ratios to rank regressors in order
of importance, in order to select a subset of model $\Gamma _{s}$. In this
study we propose to complement 
the $t$-ratios with a GSA measure, called the `total sensitivity index'. An
algorithm is then developed which combines this new ranking with the ranking
by $t$-statistics; we call this the `GSA algorithm'. Following HP, we define
a testing sequence based on this new ranking. Unlike in HP, we adopt a
`bottom-up' selection process which builds candidate models by adding
regressors in descending order of importance; this `bottom-up' selection
process has better theoretical properties, see e.g. \citet{Paruolo2001b},
and can still be interpreted as a GETS procedure. In this section we
introduce the total sensitivity index; the description of the GSA algorithm
is deferred to Section \ref{sec:algorithm}.

The total sensitivity index in GSA is based on systematic MC exploration of
the space of the inputs, as is commonly practiced in mathematical modeling
in natural sciences and engineering. The `mechanistic' models in these
disciplines are mostly principle-based, possibly involving the solution of
some kind of (differential) equation or optimization problem, and the output
- being the result of a deterministic calculation - does not customarily
include an error term. Reviews of global sensitivity analysis methods used
therein are given in \cite{helton06,santner03,saltelli2012ChemRev}.\footnote{%
A recent application of these methods to the quality of composite indicators
is given in \cite{Paruolo2013}.}

These techniques are applied here conditionally on the sample $\mZ=(\vy,\mX)$
generated as in Section \ref{subsec_HPworld}. Conditionally on the sample $%
\mZ$, we consider a measure of model fit, such as an information criterion,
indicated as $q\left( \vgamma\right) $. In the application we take $q\left( %
\vgamma\right) $ to be BIC.\footnote{$q(\vgamma )$ is a function of $\mZ$,
but we omit to indicate this in the notation.
} Remark that $q$ is a continuous random variable that depends on the
discretely-valued $\vgamma$ (conditionally on $\mZ$).

GSA\ aims to explore the effect on $q$ when varying inputs $\vgamma$ across
the full hyperspace $\Gamma $ of all possible $\vgamma$ configurations.
Specifically, we employ the total sensitivity index of each variable; this
can be interpreted as a measure of how much the given regressor contributes
to the model fit, as represented by any likelihood-related quantity, such as
BIC.

This sensitivity measure belongs to the class of variance-based methods,
which aim to decompose the variance of $q$ across models into portions
attributable to inputs and sets of inputs. Variance-based measures of
sensitivity are the computer-experiment equivalents of the analysis of the
variance of in experimental design, see \cite{archer97}.\footnote{%
In experimental design the effects of factors are estimated over levels;
instead, variance-based methods explore the entire distribution of each
factor.} The objective is to capture both the main effect and the
interaction effects of the input factors onto the output $q$, see \cite%
{saltelli2012ChemRev}.

\subsection{Sensitivity Measures}

Let $\mathcal{Q}=\{q\left( \vgamma\right) ,\vgamma\in \Gamma \}$ be the set
of all possible values of $q$ varying $\vgamma$ in the set of models $\Gamma
$. Let $P$ be the uniform distribution on $\vgamma\in \Gamma $ and $\mathbb{P%
}$ be the induced probability measure on $q$; let also $\mathbb{E}$, $%
\mathbb{V}$ indicate the expectation and variance operator with respect to $%
\mathbb{P}$. This probability space on the model space is introduced here
only to simplify exposition, and it does not correspond to any ex-ante
probability on the space of models.

We next partition the $\vgamma$ vector into two components $\gamma _{i}$ and
$\vgamma_{-i}$, where $\vgamma_{-i}$ contains all elements in $\vgamma$
except $\gamma _{i}$. We let $\mathbb{E}_{a}(\cdot |b)$ and $\ \mathbb{V}%
_{a}(\cdot |b)$ (respectively $\mathbb{E}_{a}(\cdot )$ and $\ \mathbb{V}%
_{a}(\cdot )$)\ indicate the conditional (respectively marginal) expectation
and variance operators with respect to a partition $(a,b)$ of $\vgamma$,
where $a$ and $b$ are taken equal to $\gamma _{i}$ and to $\vgamma_{-i}$.

Variance-based measures rely on decomposing the variance of the output, $V=%
\mathbb{V}(q)$, into portions attributable to inputs and sets of inputs.
There are two commonly-accepted variance-based measures, the `first-order
sensitivity index' $S_{i}$, \cite{sobol93}, and the `total-order sensitivity
index' $S_{Ti}$, \cite{homma96}.

The first-order index measures the contribution to $V$ of varying the $i$-th
input alone, and it is defined as $S_{i}={\mathbb{V}_{\gamma_{i}}\left(
\mathbb{E}_{\vgamma _{-i}}\left( q\mid \gamma_{i}\right) \right) } / {V}$.
This corresponds to seeing the effect of including or not including a
regressor, but averaged over all possible combinations of other regressors.
This measure does not account for interactions with the inclusion/exclusion
of other regressors; hence it is not used in the present paper.

Instead, we focus here on the total effect index,
which is defined by \cite{homma96} as%
\begin{equation}
S_{Ti}=\frac{\mathbb{E}_{\vgamma_{-i}}\left( \mathbb{V}_{\gamma _{i}}\left(
q\mid \vgamma_{-i}\right) \right) }{V}=1-\frac{\mathbb{V}_{\vgamma%
_{-i}}\left( \mathbb{E}_{\gamma _{i}}\left( q~|~\vgamma_{-i}\right) \right)
}{V}.  \label{eq:STi}
\end{equation}%
In the following we indicate the numerator of $S_{Ti}$ as $\sigma _{Ti}^{2}=%
\mathbb{E}_{\vgamma_{-i}}\left( \mathbb{V}_{\gamma _{i}}\left( q\mid \vgamma%
_{-i}\right) \right) $, and we use the shorthand $S_{T}$ for $S_{Ti}$.

Examining $\sigma _{Ti}^{2}$, one can notice that the inner term, $\mathbb{V}%
_{\gamma_{i}}\left( q\mid \vgamma_{-i}\right) $, is the variance of $q$ due
inclusion/exclusion of regressor $i$, but conditional on a given combination
$\vgamma_{-i}$ of the remaining regressors.
The outer expectation then averages over all values of $\vgamma_{-i}$; this
quantity is then standardised by $V$ to give the fraction of total output
variance caused by the inclusion of $x_{i}$. The second expression shows
that $S_{Ti}$ is 1 minus the first order effect for $\vgamma_{-i}$.

These measures are based on the standard variance decomposition formula, or
`law of total variance', see e.g. \cite{Billingsley1995}, Problem 34.10(b).
In the context of GSA these decomposition formulae are discussed in \cite%
{archer97}, \cite{saltelli02rel}, \cite{sobol93}, \cite{rabitz2010}. For
further reading about GSA\ in their original setting, we refer to \cite%
{saltelli2012ChemRev}.

\subsection{Monte Carlo Estimation}

In order to calculate the total sensitivity measure $S_{Ti}$ one should be
able to compute $q(\vgamma)$ for all $\vgamma\in \Gamma $, which is
unfeasible or undesirable. Instead, $S_{Ti}$ can be estimated by MC,
sampling from the space of inputs $\Gamma $. Here we select $\gamma _{i}\in
\{0,1\}$ with $P(\gamma _{i}=0)=P(\gamma _{i}=1)=0.5$, independently of $%
\vgamma_{-i}$.

This suggests the following MC sampling scheme. Generate a random draw of $%
\vgamma$ from $P$, say $\vgamma_{\ast }$; then consider elements $\vgamma%
_{\ast }^{(i)}$ with all elements equal to $\vgamma_{\ast }$ except for the $%
i$-th coordinate which is switched from 0 to 1 or vice-versa, $\gamma _{\ast
i}^{(i)}=1-\gamma _{\ast i}$. Doing this for each coordinate $i$ generates $%
p $ additional points $\vgamma _{\ast }^{(i)}$, and $p$ pairs of $\vgamma$
vectors, $\vgamma_{\ast }$ and $\vgamma_{\ast }^{(i)}$, that differ only in
the coordinate $i$. This is then used to calculate $g(\vgamma)$ and apply an
ANOVA-like estimation of main effect and residual effects.

More precisely, initialize $\ell$ at $1$; then:

\begin{enumerate}
\item Generate a draw of $\vgamma $ from $P$, where $\vgamma$ is a $p$%
-length vector with each element is randomly selected from $\{0,1\}$. Denote
this by $\vgamma_{\ell}$.

\item Evaluate $q_{\ell}=q(\vgamma_{\ell})$.

\item Take the $i$th element of $\vgamma_{\ell}$, and invert it, i.e.\ set
it to 0 if it is 1, and 1 if it is 0. Denote this new vector with inverted $%
i $th element as $\vgamma_{\ell}^{(i)}$.

\item Evaluate $q_{i\ell}=q(\vgamma_{\ell}^{(i)})$.

\item Repeat steps 3 and 4 for $i=1,2,...,p$.

\item Repeat steps 1-5 $N$ times, i.e.\ for $\ell =1,2,...,N$.
\end{enumerate}

The computational MC estimator for $\sigma _{Ti}^{2}$ and $V~$are defined as
follows, see \cite{saltelli10}, 
\begin{equation}
\hat{\sigma}_{Ti}^{2}=\frac{1}{4N}\sum_{\ell =1}^{N}\left( q_{i\ell
}-q_{\ell }\right) ^{2},\qquad \hat{V}=\frac{1}{N-1}\sum_{\ell =1}^{N}\left(
q_{\ell }-\bar{q}\right) ^{2},  \label{eq:estSTi}
\end{equation}%
where $\bar{q}=\frac{1}{N}\sum_{\ell =1}^{N}q_{\ell }$. This delivers the
following plug-in estimator for $S_{T}$, $\hat{S}_{Ti}=\hat{\sigma}_{Ti}^{2}/%
\hat{V}$. 
Readers familiar with sensitivity analysis may notice that the estimator in (%
\ref{eq:estSTi}) is different by a factor of 2 to the estimator quoted in
\cite{saltelli10}. The reason for this is given in eq. \eqref{eq_STi_BIC} in
Appendix A.\footnote{%
A heuristic reason for this is that the MC method involves a probabilistic
exploration of models, and $P(\gamma _{i}=0)=P(\gamma _{i}=1)=0.5$. Note
that in MC analyses with continuous variables, it is usually advisable to
use low-discrepancy sequences due to their space-filling properties, see
\cite{sobol67}, which give faster convergence with increasing $N$. However,
since $\vgamma$ can only take binary values for each element,
low-discrepancy sequences offer no obvious advantage over (pseudo-)random
numbers.}

We investigate the theoretical properties of ordering of variables based on $%
S_{T}$ in Appendix A; there we show that these orderings satisfy the
following minimal requirement. When the true regressors included in the DGP
and the irrelevant ones are uncorrelated, the ordering of regressors based
on $S_{T}$ separates the true from the irrelevant regressors in large
samples. One may hence expect this result to apply to other more
general situations.\footnote{%
The results in Appendix A show that one can build examples where the
ordering of regressors based on $S_{T}$ fails to separate the sets of true
and irrelevant regressors. The following analysis investigates how often
this happens in HP's designs of experiments.}

We also investigated the contribution of $S_{T}$ in practice, using HP's
experimental designs, as reported in the following section.

\subsection{Ordering variables}

\label{sec:motivation}In order to see how $S_{T}$ can be used as an
alternative or complmentary method of ranking regressors, the following
numerical experiment was performed. For each of the 11 DGPs under
investigation, $N_{R}=500$ samples $\mZ$ were drawn; on each sample,
regressors were ranked by the $t$-test and $S_{T}$, using $N=128$ in (\ref%
{eq:estSTi}). We generically indicate method $m$ for the ordering, where $m$
takes the values $t$ for $t$-test and $S$ for $S_{T}$ orderings. Both for
the $t$-test ranking and the $S_{T}$ ranking, the ordering is from the
best-fitting regressor to the worst-fitting one.

In order to measure how successful the two methods were in ranking
regressors, we defined the following measure $\delta $ of minimum relative
covering size. Indicate by $\varphi _{0}=\{i_{1},\dots ,i_{r_{0}}\}$ the set
containing the positions $i_{j}$ of the true regressors in the list $%
i=1,\dots ,p$; i.e. for each $j$ one has $\gamma _{0,i_{j}}=1$. Recall also
that $r_{0}$ is the number of elements in $\varphi _{0}$. Next, for a
generic replication $j$, let $\varphi _{\ell }^{(m)}=\{i_{1}^{(m)},\dots
,i_{\ell }^{(m)}\}$ be the set containing the first $\ell $ positions $%
i_{j}^{(m)}$ induced by the ordering of method $m$.
Let $b_{j}^{(m)}=\min \{\ell :\varphi _{0}\subseteq \varphi _{\ell }^{(m)}\}$
be the minimum number of elements $\ell $ for which $\varphi _{\ell }^{(m)}$
contains the true regressors. \ We observe that $b_{j}^{(m)}$ is well
defined, because at least for $\ell =p$ one always has $\varphi
_{0}\subseteq \varphi _{p}^{(m)}=\{1,\dots ,p\}$. We define $\delta $ to
equal $b_{j}^{(m)}$ divided by its minimum; this corresponds to the
(relative) minimum number of elements in the ordering $m$ that covers the
set of true regressors.

Observe that, by construction, one has $r_{0}\leq b_{j}^{(m)}\leq p$, and
that ideally one wishes $b_{j}^{(m)}$ to be as small as possible; ideally
one would like to have to have $b_{j}^{(m)}=r_{0}$. Hence for $\delta
_{j}^{(m)}$ defined as $b_{j}^{(m)}/r_{0}$ one has $1\leq \delta
_{j}^{(m)}\leq p/r_{0}$. \ We then compute $\delta ^{(m)}$ as the average $%
\delta _{j}^{(m)}$ over $j=1,\dots ,N_{R}$, i.e. $\delta ^{(m)}=\frac{1}{%
N_{R}}\sum_{j=1}^{N_{R}}\delta _{j}^{(m)}$.

For example, if the regressors, ranked in descending order of importance by
method $m$ in replication $j$, were $x_{3}$, $x_{12}$, $x_{21}$, $x_{11}$, $%
x_{4}$, $x_{31},...$, and the true DGP were $x_{3}$, $x_{11}$ the measure $%
\delta _{j}$ would be 2; in fact the smallest-ranked set containing $x_{3}$,
$x_{11}$ has 4 elements $b_{j}^{(m)}=4$, and $r_{0}=2$.

\begin{table}[tbp] \centering%
\begin{tabular}{l|rrrrrrrrrrr|r}
\hline
DGP & 1 & 2 & 3 & 4 & 5 & 6 & 6A & 6B & 7 & 8 & 9 & Mean \\ \hline
$S_{T}$ &  & 1.00 & 1.01 & 1.00 & 1.00 & 1.00 & 1.12 & 1.02 & \textbf{1.15}
& \textbf{1.64} & \textbf{1.13} & 1.11 \\
$t$-test &  & 1.00 & \textbf{1.53} & \textbf{1.04} & 1.00 & \textbf{1.06} &
\textbf{3.95} & \textbf{1.14} & 1.04 & 1.00 & 1.01 & \textbf{1.38} \\ \hline
\end{tabular}%
\caption{Values of $\protect\delta$ for all DGPs (average over 500 data
replications per DGP), using t-test and $S_{T}$. Mean refers to average across DGPs.
Comparatively poor rankings are in boldface.}\label{table_ST_vs_t-test}%
\end{table}%

The results over the $N_{R}=500$ replications are summarized in Table \ref%
{table_ST_vs_t-test}. Overall $S_{T}$ appears to perform better than $t$%
-ordering. 
For some DGPs (such as DGP 2 and 5) both approaches perform well ($\delta =1$
indicating correct ranking for all 500 data sets). There are other DGPs
where the performance is significantly different. In particular the $t$-test
is comparatively deficient on DGPs 3 and 10, whereas $S_{T}$ performs worse
on DGP 8. This suggests that there are some DGPs in which $S_{T}$ may offer
an advantage over the $t$-test in terms of ranking regressors in order of
importance. This implies that a hybrid approach, using both measures, may
yield a more efficient method of regressor selection, which leads to the
model selection algorithm proposed in the following section.

\section{GSA algorithm}

\label{sec:algorithm} In this section we present a hybrid approach, which
combines the search paths obtained using the $t$-ratios and the $S_{T}$
measures, and then selects the best model between the two resulting
specifications. The combined procedure is expected to be able to reap the
advantages of both orderings. For simplicity, we call this algorithm the GSA
algorithm, despite the fact that it embodies some of the characteristics of
the HP algorithm. The rest of this section contains a description of the GSA
algorithm in its basic form and with two modifications.

\subsection{The basic algorithm}

The procedure involves ranking the regressors by $t$-score or $S_{T}$, then
adopting a `bottom up' approach, where candidate models are built by
successively adding regressors in order of importance. The steps are as
follows.

\begin{enumerate}
\item Order all regressors by method $m$ (i.e. either the $t$-score or $%
S_{T})$.

\item Define the initial candidate model as the empty set of regressors.

\item Add to the candidate model the highest-ranking regressor (that is not
already in the candidate model).

\item Perform an $F$ test, comparing the validity of the candidate model to
that of the GUM.

\item If the $p$-value of the $F$ test in step 4 is below a given
significance level $\alpha$, go to step 3 (continue adding regressors),
otherwise, go to step 6.

\item Since the $F$-test has not rejected the model in step 4, this is the
selected model $\vgamma^{(m)}$.
\end{enumerate}

In the following, we use the notation $\vgamma^{(t)}$ (respectively $\vgamma%
^{(S)}$) when $t$-ratios (respectively $S_{T}$) are used for the ordering.
Note that candidate variables are added starting from an empty
specification; this is hence a `bottom up' approach.

We observe that this `bottom up'\ approach is in line with the GETS
philosophy of model selection; in fact it corresponds to the nesting of
models known the `Pantula-principle' in cointegration rank determination,
see \cite{Johansen1996}. Every model in the sequence is compared with the
GUM, and hence the sequence of tests can be interpreted as an implementation
of the GETS philosophy. Moreover, it can be proved that, for large sample
sizes, the sequence selects the smallest true model in the sequence with
probability equal to $1-\alpha $, where $\alpha $ is the size of each test.
Letting $\alpha $ tend to 0 as the sample size gets large, one can prove
that this delivers a true model with probability tending to 1.\footnote{%
See for instance \cite{Paruolo2001b}. Recall that any model whose set of
regressors contains the true one is `true'.}

As a last step, the final choice of regressors $\widehat{\vgamma}$ is chosen
between $\vgamma^{(t)}$ and $\vgamma^{(S)}$ as the one with the fewest
regressors (since both models have been declared valid by the $F$-test). If
the number of regressors is the same, but the regressors are different, the
choice is made using the BIC.

The GSA algorithm depends on some key constants; the significance level of
the $F$-test, $\alpha $, is a truly `sensitive' parameter, in that varying
it strongly affects its performance. Of the remaining constants in the
algorithm, $N$, the number of points in the GSA design, can be increased to
improve accuracy; in practice it was found that $N=128$ provided good
results, and further increases made little difference.

\subsection{Adaptive-$\protect\alpha$}

Varying $\alpha $ essentially dictates how `strong' the effect of regressors
should be to be included in the final model, such that a high $\alpha $
value will tend to include more variables, whereas a low value will cut out
variables more harshly. The difficulty is that some DGPs require low $\alpha
$ for accurate identification of the true regressors, whereas others require
higher values. Hence, there could exist no single value of $\alpha $ that is
suitable for the identification of all DGPs.

A proposed modification to deal with this problem is to use an `adaptive-$%
\alpha $' , $\alpha _{\phi}$, which is allowed to vary depending on the
data. This is based on the observation that the $F$-test returns a high $p$%
-value $p_{\text{H}}$ (typically of the order 0.2-0.6) when the proposed
model is a superset of the DGP, but when one or more of the true regressors
are missing from the proposed model, the p-value will generally be low, $p_{%
\text{L}}$ (of the order $10^{-3}$ say). The values of $p_{\text{H}} $ and $%
p_{\text{L}}$ will vary depending on the DGP and data set, making it
difficult to find a single value of $\alpha $ which will yield good results
across all DGPs. For a given DGP and data set, the $p_{\text{H}}$ and $p_{%
\text{L}}$ values are easy to identify.

Therefore, it is proposed to use a value of $\alpha _{\phi }$, such that for
each data set,
\begin{equation}
\alpha _{\phi }=p_{\text{L}}+\phi (p_{\text{H}}-p_{\text{L}})
\label{eq:adapt-alpha}
\end{equation}%
where $p_{\text{H}}$ is taken as the $p$-value resulting from considering a
candidate model with all regressors that have $S_{Ti}>0.01$ against the GUM,
and $p_{\text{L}}$ is taken as the $p$-value from considering the empty set
of regressors against the GUM. The reasoning behind the definition of $p_{%
\text{H}}$ is that it represents a candidate model which will contain the
DGP regressors with a high degree of confidence. $\phi $ is a tuning
parameter that essentially determines how far between $p_{\text{L}}$ and $p_{%
\text{H}}$ the cutoff should be. Figure \ref{fig:p_values_m11(6A)}
illustrates this on a data set sampled from DGP 6B. Note that $\alpha _{\phi
}$ is used in the $F$-test for both the $t$-ranked regressors as well as
those ordered by $S_{T}$.

\begin{figure}[tbp]
\centering
\includegraphics[width=0.75\linewidth]{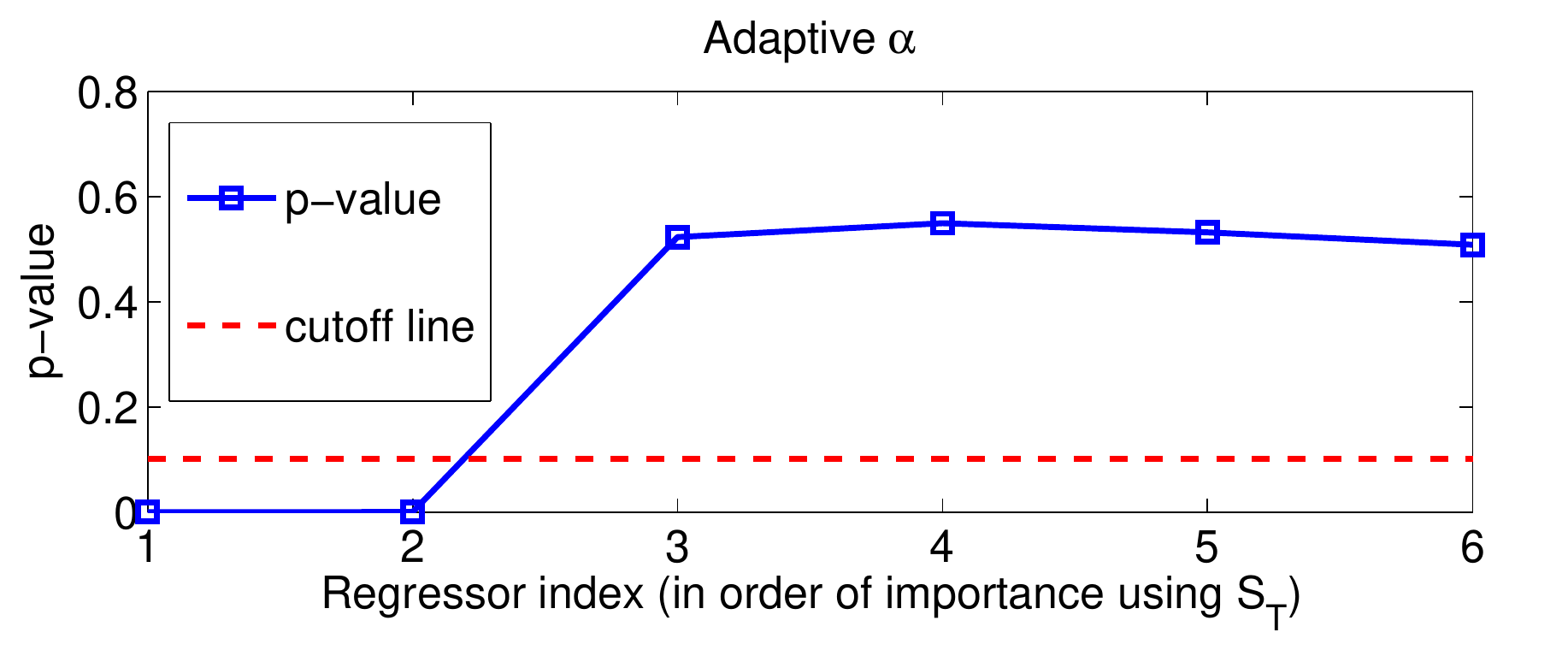}
\caption{$p$-values from $F$-test comparing candidate models to the GUM in a
sample from DGP 6B, for the 6 highest-ranked regressors. Here $\protect\phi %
=0.2$ and $\protect\alpha _{\protect\phi}$ is marked as a dotted line.}
\label{fig:p_values_m11(6A)}
\end{figure}

\subsection{Skipping regressors}

In order to correct situations where the ordering of the regressors is not
the correct one, we present here an extension of the algorithm that allows
the possibility to skip regressors in the final model. More precisely, when
step 6 is reached, it is allowed to try removing any of the other remaining
regressors, one by one, regardless of the ranking. This approach is used
instead of an exhaustive search of the combinations of remaining regressors,
because occasionally there may still be too many regressors left to make
this feasible.

In Section \ref{sec:results} the performance of the algorithm is examined,
with and without the extensions mentioned. We call $S_{T\text{full}}$ the $%
S_{T}$ procedure with adaptive-$\alpha $ and skipping regressors;\ we call $%
S_{T\text{no-skip}}$ the same procedure without skipping regressors; we call
$S_{T\text{simple}}$ the one without adaptive-$\alpha $ and without skipping
regressors.

\section{Results}

\label{sec:results}In this section we present results, using the performance
measures introduced in Section \ref{subsec_measures}. The performance is
measured with respect to the true DGP or with respect to the effective DGP
(EDGP) that one can hope to recover, given the signal to noise ratio.
Because the HP and GSA algorithms depend on tunable constants, we give
results for various values of these constants.

The procedure employed to define the EGDP is discussed in Appendix B; it
implies that the only EDGP differing from the true DGP are DGP 6 and DGP 9.
DGP 6 contains regressors $3$ and $11$, but regressor 3 is weak and EDGP 6
hence contains only regressor 11. DGP 9 contains regressors $3$, $11$, 21 29
and 37 but regressor 3 and 21 are weak and they are dropped from the
corresponding EDGP 9. More details are given in Appendix B.

The HP algorithm depends on the significance levels $\alpha $, and the GSA
algorithm on the threshold $\phi $ (which controls $\alpha _{\phi}$) for $%
S_{T\text{no-skip}} $ and $S_{T\text{full}}$ and on $\alpha $ for $S_{T\text{%
simple}}$. Because the values of $\alpha$ and $\phi$ can seriously affect
the performance of the algorithms, a fair comparison of the performance of
the algorithms may be difficult, especially since the true parameter values
will not be known in practice. To deal with this problem, the performance of
the algorithms was measured at a number of parameter values within a
plausible range.

This allowed two ways of comparing the algorithms: first, the `optimised'
performance, corresponding to the value of $\alpha $ or $\phi $ that
produced the highest $C_{1}$ score, averaged over the 11 DGPs. This can be
viewed as the `potential performance'. In practice, the optimization was
performed with a grid search on $\alpha $ and $\phi $ with
$N_{R}=10^{3}$ replications, averaging across DGPs.

Secondly, a qualitative comparison was drawn between the algorithms of the
average performance over the range of parameter values. This latter
comparison gives some insight into the more realistic situation where the
optimum parameter values are not known.

\begin{table}[tbp] \centering%
\begin{tabular}{l|rrr|rrr|rrr|rrr}
\hline
EDGP & \multicolumn{3}{|c}{$S_{T\text{simple}}$} & \multicolumn{3}{|c}{$S_{T%
\text{no-skip}}$} & \multicolumn{3}{|c}{$S_{T\text{full}}$} &
\multicolumn{3}{|c}{HP$_{\text{optimized}}$} \\
& \multicolumn{3}{|c}{$\alpha =0.0371$} & \multicolumn{3}{|c}{$\phi =0.3$} &
\multicolumn{3}{|c}{$\phi =0.3$} & \multicolumn{3}{|c}{$\alpha =4\cdot
10^{-4}$} \\ \hline
& $C_1$ & Gauge\  & Pot.\  & $C_1$ & Gauge\  & Pot.\  & $C_1$ & Gauge\  &
Pot.\  & $C_1$ & Gauge\  & Pot.\  \\
1 & 98.70 & 0.11 & 100 & 99.83 & 0.01 & 100 & 99.83 & 0.00 & 100 & 99.22 &
0.02 & 100 \\
2 & 98.52 & 0.09 & 99.98 & 99.37 & 0.02 & 100 & 99.37 & 0.02 & 100 & 98.94 &
0.03 & 100 \\
3 & 79.36 & 0.80 & 94.73 & 95.23 & 0.10 & 98.53 & 95.97 & 0.06 & 98.48 &
62.01 & 0.05 & 81.17 \\
4 & 98.59 & 0.09 & 100 & 99.16 & 0.03 & 100 & 99.16 & 0.02 & 100 & 99.29 &
0.02 & 99.92 \\
5 & 98.79 & 0.08 & 100 & 99.86 & 0.00 & 100 & 99.86 & 0.00 & 100 & 99.26 &
0.02 & 100 \\
6 & 98.70 & 0.09 & 99.99 & 99.22 & 0.03 & 99.99 & 99.22 & 0.02 & 99.99 &
99.19 & 0.03 & 99.84 \\
6A & 65.31 & 0.46 & 87.86 & 78.37 & 0.66 & 97.92 & 96.24 & 0.05 & 98.52 &
85.30 & 0.55 & 92.91 \\
6B & 97.61 & 0.10 & 99.99 & 98.57 & 0.04 & 99.99 & 99.37 & 0.02 & 100 & 98.38
& 0.07 & 99.52 \\
7 & 92.66 & 0.13 & 98.58 & 97.09 & 0.09 & 99.87 & 99.50 & 0.01 & 99.90 &
98.76 & 0.03 & 99.82 \\
8 & 98.44 & 0.07 & 99.97 & 99.91 & 0.00 & 100 & 99.92 & 0.00 & 100 & 99.05 &
0.03 & 100 \\
9 & 91.38 & 0.18 & 98.62 & 96.53 & 0.11 & 99.93 & 99.61 & 0.01 & 99.94 &
98.18 & 0.04 & 99.78 \\ \hline
Mean & 92.55 & 0.20 & 98.16 & 96.65 & 0.10 & 99.66 & 98.91 & 0.02 & 99.71 &
94.33 & 0.08 & 97.54 \\ \hline
\end{tabular}%
\vspace{0.3cm}
\caption{Percentage $C_1$, gauge and potency by EDGP. Optimised parameter
values used.}\label{tab:results_EDGP}%
\end{table}%

\subsection{Search for the EDGP}

Table \ref{tab:results_EDGP} shows the classification results in terms of $%
C_{1}$ matches, as well as the potency and gauge measures, for both
algorithms at their optimal parameter values. Results are shown with and
without the extensions discussed in Section \ref{sec:algorithm}, using $%
N_{R}=10^{4}$. Recovery of the true specification is here understood in the
EDGP sense.

The $C_{1}$ column measures the percentage frequency with which the
algorithms identified the EDGP. One notable fact is that the performance of
the HP algorithm can be vastly improved (compared to the results in HP)
simply by setting $\alpha $ to a better value, in this case $\alpha =4\times
10^{-4}$, compare with Table~\ref{tab:results_HP}.


The comparison shows that with the full $S_{T}$ algorithm, the correct
classification rate, averaged over all DGPs, is improved from 94.6\% to
98.9\%, {which corresponds to five-fold drop in the failure rate from 5.4\%
to 1.1\%.}


Removing the `skipping' extension, the performance falls to 97.1\%, and
further to 92.7\% without the adaptive-$\alpha $ feature. Examining the DGPs
individually, one can see that the HP algorithm performs well on all DGPs
except 3 and 6A, where $C_{1}$ is significantly lower than for the GSA
algorithm. The $S_{T}$ method, however, performs well on all the DGPs
investigated here, giving improvements over the HP algorithm of around 30
percentage points in DGP3, and about 10 percentage points in DGP 6A. It is
evident though that the adaptive-$\alpha $ and the skipping extensions
contribute significantly to the performance in those cases.

The potency and gauge measures (also in Table \ref{tab:results_EDGP}) reveal
a little more about the nature of the errors made by the algorithms. Gauge
is very low for all algorithms, due to the fact that the DGPs consist of
only a small fraction of the number of candidate regressors, as well as the
good performance of all algorithms. One can see though, that higher gauge
measures are found in DGP 6A, indicating the inclusion of irrelevant
regressors, except for the full $S_{T}$ algorithm, which has gauges of
practically zero. The potency measures show that the true regressors are
being identified nearly all the time, except for the HP algorithm on DGP3,
which removes true regressors with some positive frequency.

\begin{table}[tbp] \centering%
\begin{tabular}{l|rrr|rrr|rrr|rrr}
\hline
DGP & \multicolumn{3}{|c}{$S_{T\text{simple}}$} & \multicolumn{3}{|c}{$S_{T%
\text{no-skip}}$} & \multicolumn{3}{|c}{$S_{T\text{full}}$} &
\multicolumn{3}{|c}{HP} \\
& \multicolumn{3}{|c}{$\alpha =0.0371$} & \multicolumn{3}{|c}{$\phi =0.3$} &
\multicolumn{3}{|c}{$\phi =0.3$} & \multicolumn{3}{|c}{$\alpha =4\cdot
10^{-4}$} \\ \hline
& $C_1$ & Gauge & Pot. & $C_1$ & Gauge & Pot. & $C_1$ & Gauge & Pot. & $C_1$
& Gauge & Pot. \\ \hline
6 & 1.00 & 0.09 & 50.01 & 1.00 & 0.03 & 50.01 & 2.00 & 0.02 & 50.01 & 3.00 &
0.03 & 49.94 \\
9 & 0.00 & 0.19 & 59.20 & 0.00 & 0.11 & 59.97 & 0.00 & 0.01 & 59.97 & 0.00 &
0.04 & 59.92 \\ \hline
Mean & 0.50 & 0.14 & 54.60 & 0.50 & 0.07 & 54.99 & 1.00 & 0.02 & 54.99 & 1.50
& 0.03 & 54.93 \\ \hline
\end{tabular}%
\vspace{0.3cm}
\caption{Percentage $C_1$, gauge and potency by DGP. Optimised parameter
values used. The mean frequency is taken over all DGPs, but only the results for
DGPs 6 and 9 are shown since the remaining results are identical to Table
\protect\ref{tab:results_EDGP}.} \label{tab:results_DGP}%
\end{table}%

\subsection{Recovering the DGP}

Although it is argued here that the signal-to-noise ratio in DGPs 6 and 9 is
too low for certain regressors to be identified, it is still worth looking
at the results with respect to the true DGP, shown in Table \ref%
{tab:results_DGP}.
All algorithms failed to identify the true DGP even once out of the $10^4$
runs. This fact is reflected in the potency, which drops from $100\%$ to $%
50\%$ (DGP 6), and about $60\%$ (DGP 9). These results are mirrored in the
original results of HP. This suggest that GSA may not help when regressors
are `weak'.

\subsection{Robustness of algorithms}


\begin{figure}[tbp]
\centering
\includegraphics[width=0.75\linewidth]{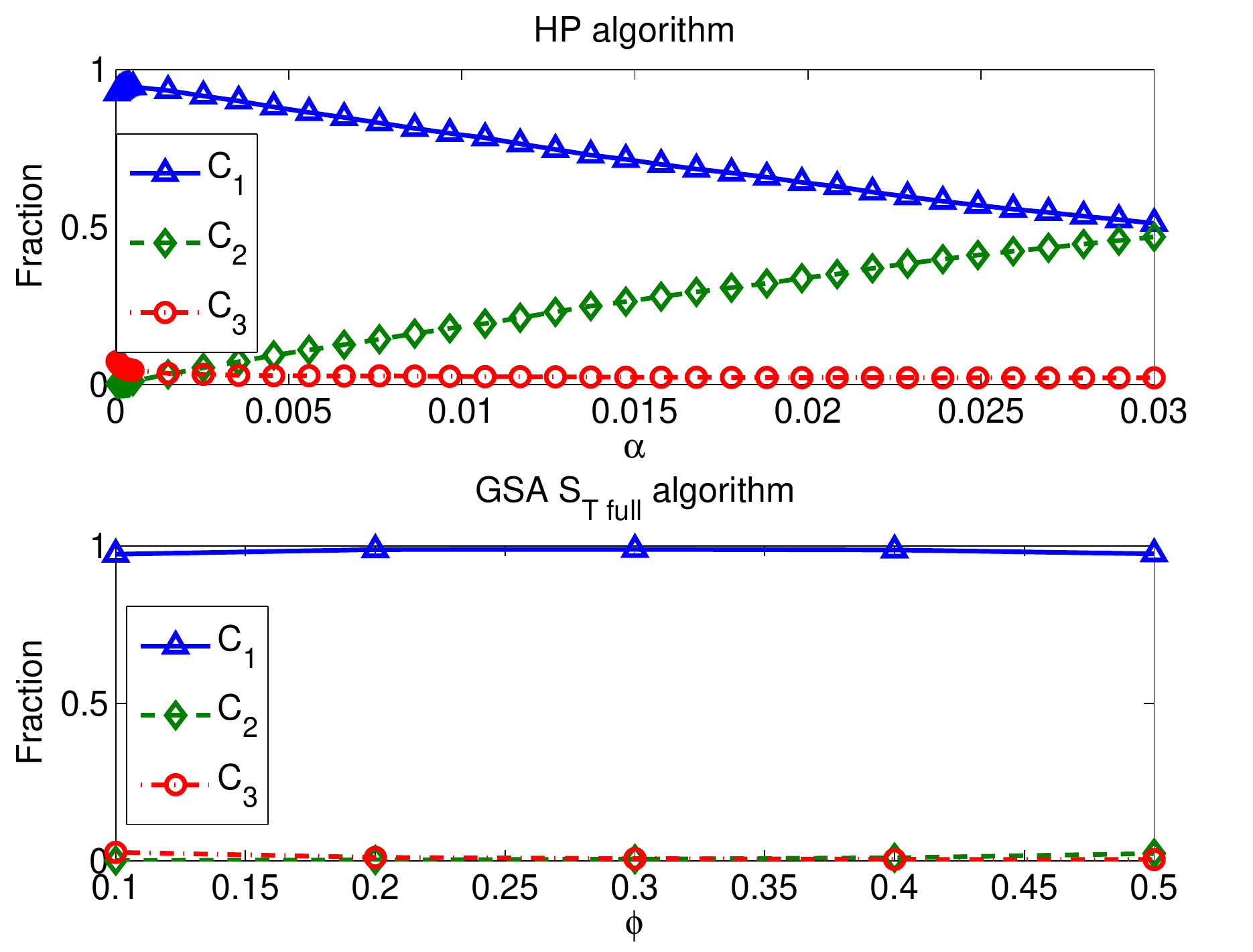}
\caption{Optimisation of algorithms with respect to tuning parameters; upper
panel: HP algorithm; lower panel: $S_{T \text{full}}$ algorithm.}
\label{fig:alg_opt}
\end{figure}


As discussed earlier, the results in Table \ref{tab:results_EDGP} are
obtained after optimisation of the tuning parameters $\alpha $ and $\phi $.
This provides a measure of potential performance, but in reality the best $%
\alpha $ and $\phi $ will not be known. For this reason it is indicative to
show the results when varying the tuning parameter. The upper panel in
Figure \ref{fig:alg_opt} shows how the categorisation of the final model
varies with $\alpha $ in the HP algorithm. It is clear that the peak
performance of the algorithm is obtained in a small neighborhood around a
rather sharp maximum at a low $\alpha $ value -- increasing $\alpha $ from
this value results in a rapid increase in $C_{2}$, whereas decreasing it
sharply increases $C_{3}$.

In contrast, the lower panel in Figure \ref{fig:alg_opt} shows the same plot
for the full $S_{T}$ algorithm. While the value of $\phi$ varies between 0.1
and 0.5, the value of $C_1$ is generally above 95\%, and $C_2$ and $C_3$ are
consistently very low. While it is difficult to accurately compare this with
the HP algorithm, due to the incomparable scales of the two optimising
parameters, the $S_{T}$ algorithm seems to be considerably more robust, and
has the advantage that the tuning parameter, $\phi$, is relatively
problem-independent.

\section{Conclusions}

\label{sec:conclusions} %

In the model selection problem, one has to choose whether or not to include
candidate variables. The approach in this paper is to view the problem as a
sensitivity analysis of a measure of fit on the space of candidate
variables.
One therefore calculates the sensitivity 
e.g. of the BIC with respect to the presence (or absence) of each candidate
variable. Since global methods are used, the global `importance' is
estimated, including all interaction effects among indicators for regressor
presence.

These interactions are in principle relevant, as the importance of a given
regressor being or not being in the DGP is conditioned by inclusion or
exclusion of the other regressors.
For this reason we used $S_{T}$, a sensitivity measure capable of
appreciating the sensitivity of a trigger for the presence of one regressor,
inclusive of its interaction effects with triggers for all other regressors.

The GSA algorithm outperforms the HP algorithm both in terms of its
potential if tuning parameters were known, and in average performance in the
practical situation when tuning parameters are unknown. The improvement is
substantial and amounts to a five-fold drop in the failure rate over the
ensemble of HP's designs. %
Arguably, the robustness of the algorithm is an {even} more distinguishing
feature, since the optimal parameter values would not be known in a
practical case.


This study has been a first exploration of new uses of GSA in the world of
model selection; it has shown to what extent measures from GSA can
contribute, albeit on a small (but representative) set of test problems.
It appears then that $S_{T}$ can be used in selecting important sources of
variation in regression. These results call for more research on the use of
GSA methods in model selection. 






\bibliographystyle{Chicago_pp}
\bibliography{EMSA_refs}

\section*{Appendix A: Properties of orderings based on $S_{Ti}$}

Recall that $S_{Ti}=\mathbb{\sigma }_{Ti}^{2}/V=\mathbb{E}_{\vgamma%
_{-i}}\left( \mathbb{V}_{\gamma _{i}}\left( q\mid \vgamma_{-i}\right)
\right) /V$. In this Appendix, we first express $\mathbb{\sigma }_{Ti}^{2}$
as a sum of terms involving $\widehat{\mathbb{\sigma }}_{\vgamma}^{2}$ for $%
\vgamma\in \Gamma$ in Lemma \ref{Lemma1}; next we show the large $n$
behavior of $\widehat{\mathbb{\sigma }}_{\vgamma}^{2}$ in Lemma \ref{Lemma2}.
Lemma \ref{Lemma2bis} states the probability limit of $\mathbb{\sigma }_{Ti}^{2}$.
Lemma \ref{Lemma3} shows that, in case the true regressors in the DGP and
the irrelevant ones are uncorrelated, $\mathbb{\sigma }_{Ti}^{2}\overset{p}{%
\rightarrow }0$ for an irrelevant regressor $i$, while $\mathbb{\sigma }%
_{Ti}^{2}\overset{p}{\rightarrow }c_{i}>0$ for a relevant one. Under the
same conditions of Lemma \ref{Lemma3}, Theorem \ref{Th_4} shows that for
large samples, a scree plot on the ordered $S_{Ti}$ allows to separate the
relevant regressors from the irrelevant ones.

Let $\gamma _{i}=\ve_{i}^{\prime }\vgamma$ and $\vgamma_{-i}=\mA_{i}^{\prime
}\vgamma$, where $\ve_{i}$ is the $i$-th column of the identity matrix of
order $p$, $\mI_{p}$ and $\mA_{i}$ is a $p\times p-1$ matrix containing all
the columns of $\mI_{p}$ except the $i$-th one. We write $q\left( \vgamma%
\right) $ as $q\left( \gamma _{i},\vgamma_{-i}\right) $ or, more simply as $%
q_{-i}\left( \gamma _{i}\right) $. Denote by $\vgamma^{(i,0)}$ the vector
corresponding to $\gamma _{i}=0$, with the remaining coordinates equal to $%
\vgamma_{-i}$, and let $\vgamma^{(i,1)}$ the vector corresponding to $\gamma
_{i}=1$ with the remaining coordinates equal to $\vgamma_{-i}$. Finally let $%
\Gamma _{-i}:=\{\vgamma_{-i}=\mA_{i}^{\prime }\vgamma,\vgamma\in \Gamma \}$.

\begin{lemma}[$\protect\sigma _{Ti}^{2}$ as an average over $\vgamma_{i}$]
\label{Lemma1}One has
\begin{equation}
\mathbb{E}_{\vgamma_{-i}}\left( \mathbb{V}_{\gamma _{i}}\left( q\mid \vgamma%
_{-i}\right) \right) =\frac{1}{4\cdot 2^{p-1}}\sum_{\vgamma_{-i}\in \Gamma
_{-i}}\left( q_{-i}\left( 1\right) -q_{-i}\left( 0\right) \right) ^{2}
\label{eq_STi_BIC}
\end{equation}%
and for $q$ equal to BIC (or any other consistent information criterion)
\begin{equation}
q_{-i}\left( 1\right) -q_{-i}\left( 0\right) =\log \left( \frac{\widehat{%
\sigma }_{\vgamma^{(i,1)}}^{2}}{\widehat{\sigma }_{\vgamma^{(i,0)}}^{2}}%
\right) +o\left( 1\right) ,  \label{eq_q1minusq0}
\end{equation}%
where $o\left( 1\right) $ is a term tending to $0$ for large $n$ and $%
\widehat{\sigma }_{\vgamma}^{2}:=n^{-1}\widehat{\vvarepsilon}_{\vgamma%
}^{\prime }\widehat{\vvarepsilon}_{\vgamma}$ where $\widehat{\vvarepsilon}_{%
\vgamma}$ are the residuals of model $\vgamma$.
\end{lemma}

\begin{proof}
Note that for $h=1,2$ one has $\mathbb{E}_{\gamma _{i}}\left( q^{h}\mid %
\vgamma_{-i}\right) =\frac{1}{2}\left( q_{-i}^{h}\left( 1\right)
+q_{-i}^{h}\left( 0\right) \right) $ so that
\begin{eqnarray*}
\mathbb{V}_{\gamma _{i}}\left( q\mid \vgamma_{-i}\right) &=&\mathbb{E}%
_{\gamma _{i}}\left( q^{2}\mid \vgamma_{-i}\right) -\left( \mathbb{E}%
_{\gamma _{i}}\left( q\mid \vgamma_{-i}\right) \right) ^{2} \\
\mathbb{\ } &=&\frac{1}{2}\left( q_{-i}^{2}\left( 1\right) +q_{-i}^{2}\left(
0\right) \right) -\frac{1}{4}\left( q_{-i}^{2}\left( 1\right)
+q_{-i}^{2}\left( 0\right) +2q_{-i}\left( 1\right) q_{-i}\left( 0\right)
\right) \\
&=&\frac{1}{4}\left( q_{-i}\left( 1\right) -q_{-i}\left( 0\right) \right)
^{2}.
\end{eqnarray*}%
Hence one finds (\ref{eq_STi_BIC}). When $q$ is BIC, $q\left( \vgamma\right)
=\log \widehat{\sigma }_{\vgamma}^{2}+k_{\vgamma}c_{n}$ with $c_{n}:=\log
(n)/n$. Other consistent information criteria replace $\log n$ with some
other increasing function $f(n)$ of $n$ with the property $%
c_{n}=f(n)/n\rightarrow 0$, see \citet{Paulsen1984} Theorem 1. Note also
that $k_{\vgamma^{(i,1)}}-k_{\vgamma^{(i,0)}}=1$, and that one has
\begin{equation}
q_{-i}\left( 1\right) -q_{-i}\left( 0\right) =\log \left( \frac{\widehat{%
\sigma }_{\vgamma^{(i,1)}}^{2}}{\widehat{\sigma }_{\vgamma^{(i,0)}}^{2}}%
\right) +\left( k_{\vgamma^{(i,1)}}-k_{\vgamma^{(i,0)}}\right) c_{n}=\log
\left( \frac{\widehat{\sigma }_{\vgamma^{(i,1)}}^{2}}{\widehat{\sigma }_{%
\vgamma^{(i,0)}}^{2}}\right) +c_{n}.  \notag
\end{equation}%
Because $c_{n}\rightarrow 0$, one finds (\ref{eq_q1minusq0}).
\end{proof}

We next wish to discuss the asymptotic behaviour of $\widehat{\sigma }_{%
\vgamma}^{2}$. Let $w_{t}:=(y_{t},x_{1,t},\dots ,x_{p,t},\varepsilon
_{t})^{\prime }$, where, without loss of generality, we assume that all
variables have mean zero. Denote $\mSigma:=E(w_{t}w_{t}^{\prime })$, where%
\begin{equation*}
\mSigma=\left(
\begin{array}{ccc}
\Sigma _{yy} & \mSigma_{y\vx} & \sigma ^{2} \\
& \mSigma_{\vx\vx} & \vzeros \\
&  & \sigma ^{2}%
\end{array}%
\right) =\left(
\begin{array}{ccccc}
\Sigma _{yy} & \Sigma _{y1} & \dots & \Sigma _{yp} & \sigma ^{2} \\
& \Sigma _{11} &  & \Sigma _{1p} & 0 \\
&  & \ddots &  &  \\
&  &  & \Sigma _{pp} & 0 \\
&  &  &  & \sigma ^{2}%
\end{array}%
\right) .
\end{equation*}
Let $\Sigma _{ij.v}:=\Sigma _{ij}-\Sigma _{iv}\Sigma _{vv}^{-1}\Sigma _{vj}$
indicate partial covariances, where $v:=\{i_{1},\dots ,i_{s}\}$ indicates a
set of indices. Note that $\mSigma_{\vx\varepsilon }=\vzeros$.

Let $\mathbb{J}:=\{1,\dots ,p\}$ be the set of the first $p$ integers, $%
\mathbb{T}:=\{i\in \mathbb{J}:\beta _{0,i}\neq 0\}$, the set of all
regressor indices in the DGP, with $r_{0}$ elements, and $\mathbb{M}:=%
\mathbb{J}\backslash \mathbb{T}$ the set of all regressor indices for
irrelevant regressors.\footnote{%
Here $\mathbb{J}\backslash \mathbb{T}$ denotes the set difference $\mathbb{J}%
\backslash \mathbb{T}:=\{i:i\in \mathbb{J},i\notin \mathbb{T}\}$;
sums over empty sets are understood to be equal to 0.} For each $%
\vgamma$, let $a_{\vgamma}$ $:=\{i_{1},\dots ,i_{k_{\vgamma}}\}^{\prime }$
indicate the set of indices $i_{j}\ $such that $\gamma _{i_{j}}=1$ in $%
\vgamma$.\ Similarly let $b_{\vgamma}:=\{i_{1},\dots ,i_{s}\}^{\prime }$
indicate the set of indices $i_{j}\ $that belong to $a_{\vgamma}\backslash
\mathbb{T}$.

We represent $\vbeta_{0}$ as $\vbeta_{0}=\mH\vphi$, where $\mH$ contains the
$r_{0}$ columns of $\mI_{p}$ corresponding to $\beta _{0,i}\neq 0$, and $%
\vphi$ contains the corresponding $\beta _{0,i}$ coefficients. Moreover we
write the matrix of regressors in the $\vgamma$ specification as $\mX\mU_{%
\vgamma}$, where $\mU_{\vgamma}$ contains the columns of $\mI_{p}$ with
column indices $a_{\vgamma}$. Define also $\mM_{\vgamma}:=\mI_{n}-\mX\mU_{%
\vgamma}\left( \mU_{\vgamma}^{\prime }\mX^{\prime }\mX\mU_{\vgamma}\right)
^{-1}\mU_{\vgamma}^{\prime }\mX^{\prime }$.

\begin{lemma}[Large sample behavior of $\protect\widehat{\protect\sigma }_{
\vgamma}^{2}$]
\label{Lemma2}As $n\rightarrow \infty $, one has
\begin{equation}
\widehat{\sigma }_{\vgamma}^{2}\overset{p}{\rightarrow }\sigma
^{2}+\sum_{h,j\in \mathbb{T}\backslash a_{\vgamma}\ }\beta _{0,h}\Sigma
_{hj.b_{\vgamma}}\beta _{0,j},  \label{eq_sigma2_plim}
\end{equation}%
where $\mathbb{T}\backslash a_{\vgamma}$ is the set of indices of the true
regressors omitted from the $\vgamma$ specification, and $b_{\vgamma}$ is
the set of indices $a_{\vgamma}\backslash \mathbb{T}$ of the regressors
included in the $\vgamma$ specification except the ones that belong to the
DGP. Remark that the sum in $(\ref{eq_sigma2_plim})$ is equal to $0$ when $%
\vgamma$ is correctly specified\ (i.e. it contains all regressors in the
DGP)\thinspace\ i.e. $\mathbb{T}\backslash a_{\vgamma}=\varnothing $.
\end{lemma}

\begin{proof}
Because $\vy=\mX\mH\vphi+\vvarepsilon$ one has
\begin{equation*}
\widehat{\sigma }_{\vgamma}^{2}=n^{-1}\vy^{\prime }\mM_{\vgamma}\vy=n^{-1}%
\vvarepsilon^{\prime }\mM_{\vgamma}\vvarepsilon+2n^{-1}\vvarepsilon^{\prime }%
\mM_{\vgamma}\mX\mH\vphi+n^{-1}\vphi^{\prime }\mH^{\prime }\mX^{\prime }\mM_{%
\vgamma}\mX\mH\vphi
\end{equation*}%
Because $\mSigma_{\vx\varepsilon }=\vzeros$, by the law or large numbers for
stationary linear processes, see \citet{Anderson1971}, one finds%
\begin{align*}
n^{-1}\vvarepsilon^{\prime }\mM_{\vgamma}\vvarepsilon& \overset{p}{%
\rightarrow }\sigma ^{2}-\mSigma_{\varepsilon \vx}\mU_{\vgamma}(\mU_{\vgamma%
}^{\prime }\mSigma_{\vx\vx}\mU_{\vgamma})^{-1}\mU_{\vgamma}^{\prime }\mSigma%
_{\vx\varepsilon }=\sigma ^{2}, \\
n^{-1}\vvarepsilon^{\prime }\mM_{\vgamma}\mX& \overset{p}{\rightarrow }%
\mSigma_{\varepsilon \vx}\left( \mI_{p}-\mU_{\vgamma}(\mU_{\vgamma}^{\prime }%
\mSigma_{\vx\vx}\mU_{\vgamma})^{-1}\mU_{\vgamma}^{\prime }\mSigma_{\vx\vx%
}\right) =\vzeros.
\end{align*}

Similarly
\begin{align*}
n^{-1}\mH^{\prime }\mX^{\prime }\mM_{\vgamma}\mX\mH& \overset{p}{\rightarrow
}\mH^{\prime }\left( \mSigma_{\vx\vx}-\mSigma_{\vx\vx}\mU_{\vgamma}(\mU_{%
\vgamma}^{\prime }\mSigma_{\vx\vx}\mU_{\vgamma})^{-1}\mU_{\vgamma}^{\prime }%
\mSigma_{\vx\vx}\right) \mH \\
& =\mH^{\prime }\mV_{\vgamma}\left( \mV_{\vgamma}^{\prime }\mSigma_{\vx\vx%
}^{-1}\mV_{\vgamma}\right) ^{-1}\mV_{\vgamma}^{\prime }\mH
\end{align*}%
where $\mV_{\vgamma}=\mU_{\vgamma,\perp }$ contains the columns in $\mI_{p}$
not contained in $\mU_{\vgamma}$, and the last equality is a special case of
a non-orthogonal projection identity, see e.g. eq. (2.13) in %
\citet{Paruolo1999} and references therein. Here $\mU_{\perp }$
indicates a basis of the orthogonal complement of the space spanned by the
columns in $\mU$. Observe that the $(p-k_{\vgamma})\times r_{0}$ matrix $\mC%
_{\vgamma}:=\mV_{\vgamma}^{\prime }\mH$ contains the columns of $\mI_{p-r_{%
\vgamma}}$ corresponding to the index set of regressors in $v_{\vgamma}:=%
\mathbb{T}\backslash a_{\vgamma}$. Hence, using e.g. eq. (A.4) in %
\citet{Paruolo1999}, one finds $\left( \mV_{\vgamma}^{\prime }\mSigma_{\vx\vx%
}^{-1}\mV_{\vgamma}\right) ^{-1}=\mSigma_{v_{\vgamma}v_{\vgamma}.b_{\vgamma%
}} $. Substituting one finds
\begin{equation*}
n^{-1}\vphi^{\prime }\mH^{\prime }\mX^{\prime }\mM_{\vgamma}\mX\mH\vphi%
\overset{p}{\rightarrow }\vphi^{\prime }\mC_{\vgamma}^{\prime }\mSigma_{v_{%
\vgamma}v_{\vgamma}.b_{\vgamma}}\mC_{\vgamma}\vphi.
\end{equation*}

Simplifying one obtains (\ref{eq_sigma2_plim}).
\end{proof}

The above results lead to the following general formulation of the probability limit
of $\sigma _{Ti}^{2}$.

\begin{lemma}[Large sample behaviour of $\protect\sigma _{Ti}^{2}$]
\label{Lemma2bis}As $n\rightarrow \infty $ one has
\begin{equation*}
\mathbb{\sigma }_{Ti}^{2}\overset{p}{\rightarrow }\frac{1}{4\cdot 2^{p-1}}%
\sum_{\vgamma_{-i}\in \Gamma _{-i}}\log \left( \frac{\sigma
^{2}+\sum_{h,j\in \mathbb{T}\backslash a_{\vgamma^{(i,1)}}}\beta
_{0,h}\Sigma _{hj.b_{\vgamma^{(i,1)}}}\beta _{0,j}}{\sigma ^{2}+\sum_{h,j\in
\mathbb{T}\backslash a_{\vgamma^{(i,0)}}}\beta _{0,h}\Sigma _{hj.b_{\vgamma%
^{(i,0)}}}\beta _{0,j}}\right)
\end{equation*}
where $a_{\vgamma}$ is the set of indices of the
regressors in the $\vgamma$ specification, and $b_{\vgamma}:=a_{\vgamma}\backslash \mathbb{T}$
includes the indices of regressors included in the $\vgamma$ specification except the ones that belong to the
DGP.
\end{lemma}

\begin{proof}
Apply Lemma \ref{Lemma1} and \ref{Lemma2}.
\end{proof}

Lemma \ref{Lemma2bis} shows that the limit behavior of $\sigma _{Ti}^{2}$ depends on the covariance structure $\mSigma$.
Some covariance structures imply that, in the limit, the value of $S _{T}$ for true regressors is greater than the value of
$S _{T}$ for irrelevant regressors.
There also exist other covariance structures which can imply a reverse ordering.\footnote{%
Worked out examples illustrating both situations are available from the authors upon request.}
In the special case when true and irrelevant regressors are uncorrelated, the next Lemma \ref{Lemma3} shows that
$S_{T}$ converges to 0 for irrelevant regressors, while $S_{T}$ converges to a positive constant for true regressors.
This result is then used in Theorem \ref{Th_4} to show that the ordering based on $S_{T}$ separates true and irrelevant
regressors in this special case.

\begin{lemma}[Orthogonal regressors in $\mathbb{M}$ and $\mathbb{T}$]
\label{Lemma3}Assume that $\Sigma _{\ell j}=0$ for all $j\in \mathbb{T}$ and
$\ell \in \mathbb{M}$. \ Then when $i\in \mathbb{M}$ one has, as $%
n\rightarrow \infty $, $\mathbb{\sigma }_{Ti}^{2}\overset{p}{\rightarrow }0$%
, whereas otherwise when $i\in \mathbb{T}$ one finds%
\begin{equation}
\mathbb{\sigma }_{Ti}^{2}\overset{p}{\rightarrow }c_{i}>0.  \label{eq_ci}
\end{equation}
\end{lemma}

\begin{proof}
From Lemma \ref{Lemma2}, one finds%
\begin{equation}
q_{-i}\left( 1\right) -q_{-i}\left( 0\right) =\log \left( \frac{\sigma
^{2}+\sum_{h,j\in \mathbb{T}\backslash a_{\vgamma^{(i,1)}}}\beta
_{0,h}\Sigma _{hj.b_{\vgamma^{(i,1)}}}\beta _{0,j}}{\sigma ^{2}+\sum_{h,j\in
\mathbb{T}\backslash a_{\vgamma^{(i,0)}}}\beta _{0,h}\Sigma _{hj.b_{\vgamma%
^{(i,0)}}}\beta _{0,j}}\right) +o_{p}(1),  \label{eq_qidiff}
\end{equation}%
Assume that $\Sigma _{\ell j}\neq 0$ for some $j\in \mathbb{T}$ and $\ell
\in \mathbb{M}$; then for some $\vgamma_{-i}\in \Gamma _{-i}$ one has $%
\Sigma _{hj.b_{\vgamma^{(i,1)}}}\neq \Sigma _{hj.b_{\vgamma^{(i,0)}}}$ in
the numerator and denominator on the r.h.s. of (\ref{eq_qidiff});\ let $%
c\neq 1$ indicate the corresponding ratio. Hence $\left( q_{-i}\left(
1\right) -q_{-i}\left( 0\right) \right) ^{2}$ converges in probability to $%
\log ^{2}c>0$, and because the terms in $\ \mathbb{E}_{\vgamma_{-i}}\left(
\mathbb{V}_{\gamma _{i}}\left( q\mid \vgamma_{-i}\right) \right) =\frac{1}{%
4\cdot 2^{p-1}}\sum_{\vgamma_{-i}\in \Gamma _{-i}}\left( q_{-i}\left(
1\right) -q_{-i}\left( 0\right) \right) ^{2}$, see Lemma \ref{Lemma2bis}, are
non-negative, one concludes that $\mathbb{\sigma }_{Ti}^{2}\overset{p}{%
\rightarrow }c_{i}>0$.

Assume instead that $\Sigma _{\ell j}=0$ for all $j\in \mathbb{T}$ and $\ell
\in \mathbb{M}$ and $i\in \mathbb{M}$. Then $\mathbb{T}\backslash a_{\vgamma%
^{(i,\cdot )}}=\mathbb{T}\backslash a_{\vgamma_{-i}}$ and, because $\Sigma
_{\ell j}=0$ for all $j\in \mathbb{T}$ and $\ell \in \mathbb{M}$, one has $%
\Sigma _{jb_{\vgamma^{(i,\cdot )}}}=0$. This implies $\Sigma _{hj.b_{\vgamma%
^{(i,\cdot )}}}:=\Sigma _{hj}-\Sigma _{hb_{\vgamma^{(i,\cdot )}}}\Sigma _{b_{%
\vgamma^{(i,\cdot )}}b_{\vgamma^{(i,\cdot )}}}^{-1}\Sigma _{b_{\vgamma%
^{(i,\cdot )}}j}=\Sigma _{hj}$. Hence%
\begin{equation*}
q_{-i}\left( 1\right) -q_{-i}\left( 0\right) =\log \left( \frac{\sigma
^{2}+\sum_{h,j\in \mathbb{T}\backslash a_{\vgamma_{-i}}}\beta _{0,h}\Sigma
_{hj}\beta _{0,j}}{\sigma ^{2}+\sum_{h,j\in \mathbb{T}\backslash a_{\vgamma%
_{-i}}}\beta _{0,h}\Sigma _{hj}\beta _{0,j}}\right) +o_{p}(1)=o_{p}(1),
\end{equation*}
for all $\vgamma_{-i}\in \Gamma _{-i}$ because the numerator and denominator
are identical. Thus $\left( q_{-i}\left( 1\right) -q_{-i}\left( 0\right)
\right) ^{2}$ converges in probability to $\log ^{2}1=0$ for all $\vgamma%
_{-i}\in \Gamma _{-i}$, and this implies $\mathbb{\sigma }_{Ti}^{2}\overset{p%
}{\rightarrow }0$.
\end{proof}

The following theorem shows that for large samples, a scree plot on the
ordered $S_{Ti}$ allows to separate the relevant regressors from the
irrelevant ones when true and irrelevant regressors are uncorrelated.

\begin{theorem}[Ordering based on $S_{Ti}$ works for orthogonal regressors
in $\mathbb{M}$ and $\mathbb{T}$]
\label{Th_4}Assume $\Sigma _{\ell j}=0$ for all $j\in \mathbb{T}$ and $\ell
\in \mathbb{M}$ as in Lemma $\ref{Lemma3}$. Define $(S_{T(1)},S_{T(2)},\dots
,S_{T(p)})$ as the set of $S_{Ti}$ values in decreasing order, with $%
S_{T(1)}\geq S_{T(2)}\geq \dots \geq S_{T(p)}$. Then as $n\rightarrow \infty
$ one has
\begin{equation*}
(S_{T(1)},S_{T(2)},\dots ,S_{T(p)})\overset{p}{\rightarrow }%
(c_{(1)},c_{(2)},\dots ,c_{(r_{0})},0,\dots 0)
\end{equation*}%
where $(c_{(1)},c_{(2)},\dots ,c_{(r_{0})})$ is the set of $c_{i}$ values
defined in $(\ref{eq_ci})$ in decreasing order. Hence the ordered $S_{Ti}$
values separate the block of true regressors in $\mathbb{T}$ in the first $%
r_{0}$ positions and the irrelevant ones in the last $p-r_{0}$ ones.
\end{theorem}

\begin{proof}
Direct application of Lemma $\ref{Lemma3}.$
\end{proof}

\section*{Appendix B: Effective DGP}

\label{sec:EDGP}In this appendix we describe how the notion of `weak
regressors' was made operational in the present context.
%
We employ a recently-introduced measure known
as the `Parametricness Index' (PI), \cite{liu11}, to identify
the `effective DGP'\ (EDGP).
Parametricness, in the sense of Liu and Yang, is a measure dependent both on
sample size \emph{and} the proposed model; a model is parametric if omission
of any of its variables implies a marked change in its fit, and
nonparametric otherwise.\footnote{%
For example, consider a data set generated by a sine function, with added
noise. If it is proposed to model this with a quadratic equation, the
data/model should be considered non-parametric. However, if the proposed
model included sinusoidal terms, it should be considered parametric.} Here
we take parametricness as a sign of detectability, i.e. of a sufficiently
high signal-noise ratio. We apply this concept both to complete
specifications as well as to single regressors;\ in particular we define the
EDGP as the subset of DGP regressors which the PI would classify as
parametric. Details are given in the following.

Considering a model $\vgamma_{k}\in \Gamma $, one can express the regression
fit as $\hat{\vy}_{k}=\mP_{k}\vy$, where $\mP_{k}$ is the projection matrix on $%
\col(\mX\mD_{\vgamma_{k}})$, and $\col$ indicates the column space; let $r_{%
\vgamma_{k}}$ be the dimension of $\text{col}(\mX\mD_{\vgamma_{k}})$. The
index PI is defined in terms of an information criterion $IC$, which depends
on $\lambda _{n}$, $d$ and $\hat{\sigma}^{2}$. Here $\lambda _{n}$ is a
nonnegative sequence that satisfies $\lambda _{n}\geq (\log n)^{-1}$, $d$ is
a nonnegative constant and $\hat{\sigma}^{2}$ is a consistent estimator of $%
\sigma ^{2}$ such as $||\vy-\hat{\vy}_{k}||^{2}/(n-r_{\vgamma_{k}})$ with $%
\vgamma_{k}$ consistent for $\vgamma_{0}$. In our application we used $%
\vgamma_{k}=\vgamma_{0}$. The information criterion IC is defined by
\begin{equation}
IC_{\lambda _{n},d}(\vgamma_{k},\hat{\sigma}^{2})=||\vy-\hat{\vy}%
_{k}||^{2}+\lambda _{n}\log (n)r_{k}\hat{\sigma}^{2}-n\hat{\sigma}%
^{2}+dn^{1/2}\log (n)\hat{\sigma}^{2}  \label{eq:IC}
\end{equation}%
where $||\cdot ||$ represents Euclidean distance; here we take $\lambda
_{n}=1$ and $d=0$ as suggested in \cite{liu11}.

Let now $\vgamma_{0}$ be the DGP; PI is now defined in the present context
as,
\begin{equation}
PI=\left\{
\begin{array}{cc}
\underset{\vgamma_{k}\in \Gamma _{1}(\vgamma_{0})}{\inf }~{\frac{IC_{\lambda
_{n},d}(\vgamma_{k},\hat{\sigma}^{2})}{IC_{\lambda _{n},d}(\vgamma_{0},\hat{%
\sigma}^{2})}} & \mbox{  if  }r_{{\vgamma}_{0}}>1 \\
n & \mbox{  if  }r_{\vgamma_{0}}=1%
\end{array}%
\right.  \label{eq:PIdef}
\end{equation}%
where $\Gamma _{1}(\vgamma_{0})$ is the set of submodels $\vgamma_{k}$ of
the DGP $\vgamma_{0}$ such that $r_{\vgamma_{k}}=r_{\vgamma_{0}}-1$, i.e.
all submodels obtained by removing one regressor at a time (with
replacement).\footnote{%
In the original paper $\vgamma_{0}$ is replaced by the model $\hat{\gamma}%
_{k}$ selected by a weakly consistent information criterion, such as BIC.}

The reasoning is that if the model is parametric (and correctly specified
for the data), removing any of the regressors will have a marked impact on $%
IC$. In contrast, if (some of the)\ regressors are just incremental terms in
a nonparametric approximation, removing one of these regressors will have
little effect on IC. \cite{liu11} show that PI converges to 1 for a
nonparametric scenario, and goes to infinity in a parametric scenario. The
authors suggest to take $PI=1.2$ is a cutoff point between parametric and
nonparametric scenarios; we adopt this threshold in the following.

As suggested by Liu and Yang, PI can, \textquotedblleft given the regression
function and the noise level ... indicate whether the problem is practically
parametric/nonparametric at the current sample size\textquotedblright . If
the PI value is close to or below 1, one could conclude that at least some
of the terms are `undetectable' at the given sample size, therefore it may
be unreasonable to expect an algorithm to identify the DGP correctly.

We apply PI at the level of each DGP; if PI indicates that the DGP is
nonparametric, we also investigate which of the submodel is responsible for
this and label the corresponding omitted variables as `weak'. As in the rest
of the paper, we employ a MC approach. We generated 5000 datasets from each
DGP and calculated PI for each sample, hence obtaining a distribution of PI
values. \ Table \ref{tab:PI_EDGP} summarizes the MC distribution of PI
values through the empirical distribution function $F_{m}(x)=m^{-1}%
\sum_{j=1}^{m}1(PI_{j}\leq x)$, where $m= N_R$ and $PI_{j}$ is the PI value in
replication $j=1,\dots ,N_R$. Quantiles of PI are indicated as PI$_{\alpha }$%
, with $\alpha =0.01$, $0.1$, $0.9$, $0.99$, and the MC\ mean PI is
indicated as $E_{N}$(PI), where for simplicity we drop he subscript $R$ from $N_R$.

The reference threshold is $PI=1.2$, and $F_{N}(1.2)$ shows the frequency of
PI being below this limit; in other words this gives an estimate for the DGP
to be classified as nonparametric. There is a very clear distinction: DGPs 6
and 9 are regarded as nonparametric 98\% and 100\% of the time respectively.
In contrast, all other DGPs are always regarded as parametric, with the
slight exception of DGP3, which is a little less clear cut.

Examining the quantiles, DGP3 has a mean PI value of 2.54 and $PI_{0.1}=1.53$%
, which puts it in the parametric class in the large majority of cases. DGP
6 has a mean PI of 0.53, and $PI_{0.9}=0.79$, making it almost always
nonparametric. DGP 9 has $PI_{0.99}=0.93$, making it the most obviously
nonparametric DGP. Of the remaining DGPs, all are well above the threshold
and can be safely considered parametric.

\begin{table}[tbp] \centering%
\begin{tabular}{l|l|lllllll}
\hline
DGP & DGP Indices & $F_{N}$(1.2) & PI$_{0.01}$ & PI$_{0.1}$ & $E_{N}$(PI) &
PI$_{0.9}$ & PI$_{0.99}$ & EDGP Indices \\ \hline
1 & \{\} & - & - & - & - & - & - & \{\} \\
2 & \{37\} & 0.00 & 16.55 & 25.59 & 41.80 & 60.83 & 84.61 & \{37\} \\
3 & \{37,38\} & 0.04 & 0.88 & 1.53 & 2.54 & 3.52 & 4.17 & \{37,38\} \\
4 & \{11\} & 0.00 & 30.50 & 37.82 & 49.19 & 61.78 & 74.88 & \{11\} \\
5 & \{3\} & 0.00 & 365.84 & 415.39 & 493.63 & 578.17 & 668.84 & \{3\} \\
6 & \{3,11\} & 0.98 & 0.37 & 0.38 & 0.53 & 0.79 & 1.40 & \textbf{\{11\}} \\
6A & \{3,11\} & 0.00 & 2.77 & 4.15 & 6.44 & 8.95 & 11.69 & \{3,11\} \\
6B & \{3,11\} & 0.00 & 15.11 & 18.10 & 23.04 & 28.38 & 33.72 & \{3,11\} \\
7 & \{11,29,37\} & 0.00 & 2.84 & 4.16 & 6.46 & 8.96 & 11.76 & \{11,29,37\}
\\
8 & \{3,21,37\} & 0.00 & 5.77 & 8.40 & 13.49 & 19.22 & 26.41 & \{3,21,37\}
\\
9 & \{3,11,21,29,37\} & 1.00 & 0.75 & 0.75 & 0.77 & 0.81 & 0.93 & \textbf{%
\{11,29,37\}} \\ \hline
\end{tabular}%
\caption{Distribution of PI values for DGPs 1-9. $F_{n}(\cdot )$ is the
MC cumulative distribution function of PI and $PI_{\protect\alpha }$
is the $\protect\alpha $-quantile of
$F_{m}(\cdot )$. DGPs where EDGP$\neq$DGP are in boldface.}\label%
{tab:PI_EDGP}%
\end{table}%

\begin{table}[tbp] \centering%
\begin{tabular}{l|l|llllll}
\hline
DGP & Variable & $F_{N}$(1.2) & ICR$_{0.01}$ & ICR$_{0.1}$ & $E$(ICR) & ICR$%
_{0.9}$ & ICR$_{0.99}$ \\ \hline
6 & x$_{3}$ & \textbf{0.98} & \textbf{0.37} & \textbf{0.38} & \textbf{0.53}
& \textbf{0.79} & \textbf{1.40} \\
& x$_{11}$ & 0.00 & 15.79 & 19.27 & 25.03 & 31.33 & 37.31 \\ \hline
9 & x$_{3}$ & \textbf{0.99} & \textbf{0.75} & \textbf{0.75} & \textbf{0.82}
& \textbf{0.94} & \textbf{1.20} \\
& x$_{11}$ & 0.00 & 8.44 & 9.95 & 12.56 & 15.41 & 18.36 \\
& x$_{21}$ & \textbf{0.99} & \textbf{0.75} & \textbf{0.75} & \textbf{0.81} &
\textbf{0.92} & \textbf{1.18} \\
& x$_{29}$ & 0.00 & 2.15 & 2.88 & 4.24 & 5.73 & 7.38 \\
& x$_{37}$ & 0.00 & 3.87 & 5.46 & 8.82 & 12.61 & 17.25 \\ \hline
\end{tabular}%
\caption{Distribution of ICRs for DGPs 6 and 9. Notation as in Table \protect
\ref{tab:PI_EDGP}. Variables that are excluded from the EDGP are
in boldface.}\label{tab:ICratios}%
\end{table}%

We next further investigate which regressors are causing the
nonparametricness, i.e. which regressors are `weak'. We examining the
individual IC ratios for each regressor of a given DGP, see (\ref{eq:PIdef}%
). Here we let $ICR(i)$ indicate the IC ratio between the DGP and the
submodel of the DGP where variable $i$ is removed. Table \ref{tab:ICratios}
reports the distribution of $ICR(i)$ for DGPs 6 and 9, which are the
nonparametric DGPs. One can clearly see that in DGP 6, it is $x_{3}$ that is
causing the nonparametricness, since it has a mean ICR(3) of 0.53.\ Removing
this regressor improves the information criterion given the data. The same
is true for $x_{3}$ and $x_{21}$ in DGP 9, which both have ICRs with a mean
of around 0.8. In contrast, removing any of the other regressors has a
significant impact on the quality of the model fit. In practice, therefore,
one could consider these as the weak regressors.

Therefore, in DGPs 6 and 9, the variables in boldface in Table \ref%
{tab:ICratios} are excluded from the EDGP. The EDGP are defined as the
remaining regressors in each case, see Table \ref{tab:PI_EDGP}. For
fairness, the results are presented here relative to both the DGP and the
EDGP, although it is maintained that the identification of the EDGP is a
more reasonable measure of success (a fact reflected by the results of both
algorithms, and the original work of HP).




\end{document}